\documentclass{article}

\usepackage[utf8]{inputenc} 
\usepackage[T1]{fontenc}    
\usepackage{hyperref}       
\usepackage{url}            
\usepackage{booktabs}       
\usepackage{amsfonts}       
\usepackage{nicefrac}       
\usepackage{microtype}      
\usepackage{xcolor}         
\usepackage{amsmath}        
\usepackage{amssymb}        
\usepackage{graphicx}       
\usepackage{natbib}         

\usepackage[margin=1in]{geometry}
\usepackage{setspace}
\title{Formatting Instructions For NeurIPS 2025}

\date{}

\author{%
  \begin{minipage}[t]{0.45\textwidth}
    \centering
    Zihao Fu\\
    Oxford Internet Institute\\
    University of Oxford\\
    Oxford, UK\\
    \texttt{zihao.fu@oii.ox.ac.uk}
  \end{minipage}%
  \hfill
  \begin{minipage}[t]{0.45\textwidth}
    \centering
    Chris Russell\\
    Oxford Internet Institute\\
    University of Oxford\\
    Oxford, UK\\
    \texttt{chris.russell@oii.ox.ac.uk}
  \end{minipage}
}

\input{header}
\begin{document}

\maketitle

\begin{abstract}
Digital watermarking is a promising solution for mitigating some of the risks arising from the misuse of automatically generated text. 
These approaches either embed non-specific watermarks to allow for the detection of any text generated by a particular sampler, or embed specific keys that allow the identification of the LLM user. 
However, simultaneously using the same embedding for both detection and user identification leads to a false detection problem, whereby, as user capacity grows, unwatermarked text is increasingly likely to be falsely detected as watermarked. Through theoretical analysis, we identify the underlying causes of this phenomenon.
Building on these insights, we propose Dual Watermarking which jointly encodes detection and identification watermarks into generated text, significantly reducing false positives while maintaining high detection accuracy. Our experimental results validate our theoretical findings and demonstrate the effectiveness of our approach.
\end{abstract}

\section{Introduction}
Large Language Models (LLMs) \citep{radford2018improving,radford2019language,brown2020language,achiam2023gpt,touvron2023llama,touvron2023llama2} have emerged as the dominant technology across a wide range of natural language processing tasks. While these models have significantly advanced the field, their misuse has raised numerous ethical concerns. 
Of particular concern is the use of LLMs to impersonate human text, and to generate text that appears to represent a sincere effort to engage, but is in fact automatically generated. Example use cases include the automatic generation of homework; of scientific papers, grants, and reviews; and the automatic generation of spam and astroturfing \citep{wachter2024large,yang2023poisoning,nikiforovskaya2020automatic}. The primary concern here is that such automatically generated text can flood the environment, making it impossible to identify sincere texts that are worth reading and responding to. 

In response to these harms, researchers proposed the use of digital watermarking. These watermarking techniques can be used for the \emph{detection} of LLM generated  text, i.e., confirming that a text was automatically generated by an LLM, or for \emph{identification}, identifying a particular user id or key associated with the generation of the text.

In \emph{detection}, \citet{aaronson2023watermarking,fernandez2023three} proposed distribution-based watermarking  which alter token probabilities, while \citet{kirchenbauer2023reliability,yoo2023advancing} introduced dictionary-based approaches that split the dictionary into parts. Then, these approaches detect the watermark with statistical tests.
In \emph{identification},
\citet{fernandez2023three,yoo2023advancing} demonstrated digital watermarking can identify individual users.

Identification-based watermarking carries with it the risk of privacy violations. For example, if LLMs are used as part of a pseudoanonymization process to rephrase text, it could allow the text to be traced back to the account that rephrased the text.  However, if appropriately disclosed, it could replace more intrusive governance measures. Such watermarking would allow LLM service providers to identify and shut down accounts used for generating spam without requiring the monitoring of every customer's API call. 

Detection watermarks can be formalized as binary labels that signal a text is watermarked, while identification watermarks encode specific information, such as user identity, as an integer key. Both distribution and dictionary-based watermarking can simultaneously incorporate both types of watermarks \citep{fernandez2023three,yoo2023advancing}. This enables the detection of whether a given text is generated by an LLM while simultaneously tracking user IDs. 

All existing systems that simultaneously handle both the detection and identification use the same information, and share the same shortcoming. As the key capacity (e.g. maximum supported user count) increases, unwatermarked text is increasingly likely to be wrongly detected as watermarked. We refer to this as the \emph{false detection problem}. This issue is evident in both widely used approaches: distribution-based and dictionary-based watermarking methods \cite{fernandez2023three,yoo2023advancing}. Both methods encode the entire text using an identification watermark and compare the maximum score among all candidate IDs against a threshold to get the detection watermark. However, as the number of candidates increases, this maximum operation distorted the test statistic distribution, resulting in more unwatermarked text being falsely detected as watermarked provided the False Negative Ratio (FNR) is held constant. As discussed by \citet{liu2023survey,giboulot2024watermax}, such false positives are considerably more critical than false negatives, as erroneously identifying human-generated texts as watermarked can result in more severe adverse consequences, including accusations of cheating in a scholarly context, and the wrongful suspension or shutdown of user accounts. To understand why this problem occurs, we demonstrate, both theoretically and empirically, that as the identity number increases, the false detection bound increases exponentially. This rapid increase causes detection to fail, even within a relatively small key capacity, limiting our ability to embed richer information while still using watermarks for the original purpose of confirming that a text was LLM generated.

To address this issue, we propose a Dual Watermark (DW) scheme to mitigate the problem of false detections. This scheme encodes the detection watermark and the identification watermark into distinct parts of the generated text. 
A hash decision function  determines if the current step encodes a detection or identification watermark. When detecting the watermark, we reuse the same hash decision function in conjunction with two statistics for the detection and the identification watermark respectively. This can be seamlessly incorporated into both distribution-based and dictionary-based watermarking methods. We also propose an extended Hybrid Dual Watermark (HDW) strategy, which simultaneously integrates the statistics for the indicator and the identification watermark, and compare the two approaches.

We conduct a theoretical analysis of the false positive error bounds associated with the distribution-based and dictionary-based watermarking methods, as well as our proposed DW method. Our analysis shows that at a particular FNR, the False Positive Ratio (FPR) grows uncontrollably for the distribution-based method  (\Cref{sec:bounds}) and the dictionary-based method (\Cref{sec:multibit}), when the key space -- corresponding to the number of user identities -- is large or the text length is short.
We also present theoretical analysis of our DW method, demonstrating that it significantly outperforms methods that reutilize the same statistic for the identification watermark to calculate the detection watermark (referred to as Full Key Encoding (FKE) in our paper).
Both our theoretical and empirical findings indicate that our approach substantially alleviates the false detection problem at any FNR.

Our contributions are as follows:
(1) We uncover the false detection problem in LLM watermarking and conduct a theoretical and empirical analysis to investigate it;
(2) We propose an analysis of the false detection rate, and
illustrate the potential severity of this problem under specific conditions, such as variations in text length, key space size, and other factors;
(3) We introduce two new methods DW, and HDW to reduce the false detection rate, while preserving a high true positives rate.

\section{Related Works}
As mentioned, existing watermarking methods can be categorized into two types: detection watermarks and identification watermarks. Detection watermarks are specific indicators within LLM-generated text, that indicate whether the text is generated by a watermarked LLM or not. \citet{aaronson2023watermarking,fernandez2023three,fu2024gumbelsoft} proposed to use the Gumbel trick to generate a corresponding random variable with a distinct distribution for watermarked text. \citet{kirchenbauer2023watermark,kirchenbauer2023reliability} proposed dividing the vocabulary into red and green lists based on preceding tokens. \citet{christ2024undetectable} introduced the concept of embedding undetectable watermarks in language model outputs using cryptographic techniques. \citet{zhao2023provable} proposed the Unigram-Watermark method to improve the detection accuracy and robustness of watermarks.

Identification watermarks embed much richer information within the generated text. Most existing methods embed an integer key into the generated text, which typically could be used to represent user ID. \citet{fernandez2023three} proposed utilizing a hash key to represent essential information and iterating over all possible keys to identify the key with the maximum score in detection.
\citet{yoo2023robust,yoo2023advancing,wang2023towards,boroujeny2024multi,qu2024provably} split the dictionary into several groups to represent the key IDs. \citet{abdelnabi2021adversarial} advocated for adversarially encoding information into the watermark. 

It is straightforward to simultaneously encode both indicator and data information. For distribution-based methods, we extend \citet{fernandez2023three}’s approach by comparing the maximum score across all candidate IDs with a threshold to determine the detection watermark. For dictionary-based methods, \citet{yoo2023robust,yoo2023advancing} proposed using the most frequent dictionary part in identification watermark as the indicator. However, to the best of our knowledge, no existing work acknowledges that such a combination leads to the false detection problem as the data capacity grows. We provide a thorough theoretical and empirical analysis of why this problem occurs and methods to address it.


\section{Method}

\subsection{False Detection Problem}
Existing methods detect both the detection watermark and the identification watermark simultaneously by inferring the detection watermark from the same statistics used for the identification watermark. For distribution-based methods \citep{ fernandez2023three}, this involves taking the maximum score among all possible keys and comparing it with a threshold. This leads to an increased false positive rate as the key capacity grows. We provide a quantitative theoretical analysis of this phenomenon in \Cref{sec:theory}. Similarly, dictionary-based methods \citep{yoo2023advancing} determine the dominant partition using a maximum operation. In \Cref{sec:multibit}, we show this transforms the original binomial distribution into a Gumbel distribution, resulting in a higher false positive rate that grows with the number of keys.

\subsection{Dual Watermark}\label{sec:method}

The principal cause of false detections lies in reusing the identification watermark’s statistics for detection. To address this, we propose a novel approach termed the Dual Watermark (DW). Instead of using all tokens to encode the identification watermark, we selectively use certain tokens to encode the detection watermark. Such a method avoids performing the maximum operation for both distribution-based method and the dictionary-based method, thereby mitigating the false detection problem. Additionally, we propose a Hybrid Dual Watermark (HDW) method that further leverages the identification watermark to assist in detecting the watermark.

DW is adaptable to both distribution-based methods and dictionary-based methods. To provide a clearer explanation without causing confusion, we focus on the distribution-based methods in this section and defer the discussion of the adaptation for dictionary-based methods in \Cref{sec:adaptdictionary} and  \Cref{sec:multibit}.
For distribution-based methods \citep{aaronson2023watermarking}, they deliberately construct a statistic to guide both the generation and detection strategies. Although the original paper does not explicitly identify this method as the Gumbel-Max trick \citep{gumbel1954statistical,maddison2014sampling,jang2016categorical}, this method is essentially a Gumbel-Max trick (\Cref{sec:relation}). This formulation underpins our analysis in \Cref{sec:theory}. In our discussion, we refer to all back-bone methods introduced by \citet{fernandez2023three} and \citet{yoo2023advancing} as Full Key Encoding (FKE), where all tokens are used to encode the identification watermark.


\begin{algorithm*}[h!]
\scriptsize
\caption{Watermarked Text Generation and Detection}
\label{alg:watermark}
\vspace{-2em}
\begin{multicols}{2}
\textbf{Generation Process:}
\begin{algorithmic}[1]
\REQUIRE Language model $\mathcal{L}$, key ID $\xi$, indication ratio $r_d \in [0, 1]$, token sequence $[x_1, \ldots, x_{i-1}]$
\OUTPUT Generated token $x_i$

\STATE Compute logits: $\ell_i = \mathcal{L}([x_1, \ldots, x_{i-1}])$
\STATE Compute hash key: $h_i = \mathcal{H}(x_{i-h}, \cdots, x_{i-1})$
\STATE \highlight{Determine indicator: $d_i = 1 \text{ if } (h_i \% 100) < (100 r_d) \text{ else } 0$}
\STATE \highlight{Compute salt key: $A_i = \xi \cdot \mathds{1}(d_i = 1)$}
\STATE Compute seed: $h_g = \mathcal{H}(x_{i-h}, \cdots, x_{i-1}, A_i)$
\STATE Generate uniform random variables: $u_i \sim U(0, 1, h_g)$
\STATE Transform to Gumbel variables: $g_i = -\ln(-\ln(u_i))$
\STATE Adjust logits: $\tilde{\ell}_i = \ell_i + g_i$
\STATE Next token: $x_i = \arg\max_j \tilde{\ell}_{ij}$\\
\STATE \RETURN $x_i$
\end{algorithmic}

\vfill


\textbf{Detection Process:}
\begin{algorithmic}[1]
\REQUIRE Language model $\mathcal{L}$, token sequence $[x_1, \ldots, x_T]$, sequence length $T$, candidate keys $\{\xi\}$, ratio $r_d \in [0, 1]$, thresholds $\tau_d, \tau_k$
\OUTPUT Detection watermark $I_d$, identification watermark $I_k$

\STATE Initialize indicator score: $S_d = 0$
\STATE Initialize key scores: $S_k(\xi) = 0$ for all $\xi$
\FOR{$i = 1$ to $T$}
    \STATE Compute hash key: $h_i = \mathcal{H}(x_{i-h}, \cdots,x_{i-1})$
    \STATE \highlight{Determine indicator: $d_i = \mathds{1}((h_i \% 100) < 100r_d)$}
    \FOR{each candidate key $\xi$ in $\{\xi\}$}
        \STATE Compute salt key: $A_i(\xi, d_i) = \xi \cdot \mathds{1}(d_i = 1)$
        \STATE Compute seed: $h_g(\xi) = \mathcal{H}(x_{i-2}, \cdots,x_{i-1}, A_i)$
        \STATE Generate uniform random variables: $u_i(\xi) \sim U(0, 1, h_g(\xi))$
        \STATE Update key score: $S_k(\xi) \mathrel{+}= d_i \cdot (-\ln(1 - u_{i x_i}(\xi)))$
    \ENDFOR
    \STATE \highlight{Update indicator score: $S_d \mathrel{+}= (1 - d_i) \cdot (-\ln(1 - u_{i x_i}(0)))$}
\ENDFOR
\STATE \highlight{Compute detection watermark: $I_d = \mathds{1}(S_d > \tau_d)$}
\STATE Identify identification watermark: $I_k = \arg\max_\xi S_k(\xi)$
\STATE  \RETURN $I_d, I_k$
\end{algorithmic}
\vspace{-4em}
\end{multicols}

\scriptsize \text{* \highlight{Green highlights} represent newly added components in DW compared with FKE.}

\end{algorithm*}

\subsection{Generating}
During generation, we follow \citet{aaronson2023watermarking, fernandez2023three}, using an LLM to generate a text sequence while encoding specific information through a deliberately designed sampling strategy. This strategy modifies standard stochastic sampling (simply sampling a token based on the corresponding probability) by incorporating the Gumbel-Max trick, with the seed for the Gumbel random variable being controlled by the previous tokens and the information to be encoded. We discuss the relationship between the Gumbel-Max trick and the original method \citep{aaronson2023watermarking, fernandez2023three} in \Cref{sec:relation}. The algorithm outline is provided in \Cref{alg:watermark}.

Given a key ID $\xi$ and a token sequence $[x_1, \ldots, x_{i-1}]$, where each $x_j$ is a token ID within the range \([1, \ldots, V]\), with $V$ representing the vocabulary size, an LLM $\mathcal{L}$ generates the subsequent token. This is done by calculating the logit as $\ell_i = \mathcal{L}([x_1, \ldots, x_{i-1}])$, where $\ell_i \in \mathbb{R}^V$ represents the logits for predicting the $i$th token. It then samples the tokens based on their corresponding probabilities, with techniques such as Top-k and Top-p sampling also being applicable \citep{fernandez2023three}.

To embed both the detection watermark and the identification watermark, we first calculate an indicator hash key \( h_i \) to determine if the current token encodes the detection watermark or identification watermark. The hash key \( h_i = \mathcal{H}(x_{i-h},\cdots, x_{i-1}) \) is based on the previous $h$ tokens, and a discriminative variable 
\( d_i = \mathds{1}((h_i \bmod 100) < 100r_d) \), where \( h_i \in \mathbb{N} \) is the hash key derived from tokens \( x_{i-h} \) to \( x_{i-1} \). The indication ratio parameter \( r_d \in (0,1) \) is a user-specified ratio that controls the proportion of tokens used for the detection watermark. The value of \( d_i \) can be either 0 or 1: if \( d_i = 0 \), it indicates that the current token encodes an indicator signifying that the text is watermarked; When \( d_i = 1 \), the token encodes the identification watermark. The information salt key \( A_i \) is then computed as:
$A_i(\xi,d_i) = \xi \cdot \mathds{1}(d_i = 1)$
where \( \xi \) is the user-specified identification watermark used to store keys such as a user ID. This encoding method is naturally robust to deletion or insertion, as the hash key depends solely on the previous $h$ tokens. If a small number of tokens are removed or added, most of the remaining salt keys remain unaffected, thereby preventing any significant change to the final detection score. We also empirically illustrate this claim using the insertion and deletion attack experiments described in \Cref{sec:insdel}.

Subsequently, a new seed is generated using the hash function as \( h_g = \mathcal{H}(x_{i-h}, \cdots, x_{i-1}, A_i(\xi,d_i)) \), which is then used to construct a standard uniform distribution \( u_i \sim U(0, 1, h_g) \). Here, \( u_i \in \mathbb{R}^V \) is a vector of standard uniform random variables generated with the seed \( h_g \). This uniform variable is transformed into a Gumbel variable vector \( g_i = -\ln(-\ln(u_i)) \), where \( g_i \in \mathbb{R}^V \) is a vector of standard Gumbel variables with parameters \( \mu = 0 \) and \( \beta = 1 \). Adjusted logits are then calculated as \( \tilde{\ell}_i = \ell_i + g_i \).
By incorporating the Gumbel variable \( g_i \), we alter the sampling process typically used in LLMs to directly select the token with the maximum score of \( \tilde{\ell}_i \). Consequently, the next token \( x_i \) is determined as \( x_i = \arg\max_j \tilde{\ell}_{ij} \), where \( \tilde{\ell}_{ij} \) is the \( j \)th element of the vector \( \tilde{\ell}_i \). The Gumbel trick ensures that the probability of sampling the \( k \)th token, \( P(k = \arg\max_j \tilde{\ell}_{ij}) = p_{ik} \) and thus form an unbiased estimator of original probability distribution $p_i$ \citep{fernandez2023three,liusemantic2023a}.
Thus, the sampling of next token \( x_i \) is now driven by a random uniform vector \( u_i \), generated using the previous tokens and seed \( h_g \). 

\subsection{Detecting}\label{sec:detecting}
When detecting watermarks (see \Cref{alg:watermark}), we follow the same procedure; using the previous tokens and the salt key to recover the random variable $u_{i}(\xi)\in \mathbb{R}^V$ corresponding to the current token ID $x_i$ and a probing key ID $\xi$. Here $u_{ix_i}(\xi)$ is the $x_i$-th element of the vector $u_{i}(\xi)$. If the text is not generated by the above procedure or if the salt does not match, the corresponding random variables will simply be uniformly distributed. However, if the text is generated with the specified procedure and the correct salt key, the corresponding random variable will follow a Beta distribution \citep{fernandez2023three}, as it is the maximum element of a uniform vector. Then, we use a specific test variable $S_d$ to differentiate between these two distributions.

Given an LLM $\mathcal{L}$ and token sequence $[x_1, \ldots, x_T]$, where each $x_i$ is a token ID within the range $[1, \ldots, V]$ and $T\in \mathbb{N}$ is the sequence length, we detect the detection watermark \(I_d \in \{0,1\}\) and identification watermark \(I_k \in [1,\ldots,K]\) using deliberately designed score functions. As with the generating phase, we first calculate the hash key $h_i$ based on previous tokens as $h_i = \mathcal{H}(x_{i-h}, \cdots, x_{i-1})$ and $d_i = \mathds{1}((h_i \bmod 100) < 100r_d)$, where $h_i \in \mathbb{N}, d_i \in \{0,1\}$. Unlike the generating process, here, the identification watermark $\xi$ is unknown and is inferred. We denote $A_i$ as a function of $\xi$ and $d_i$, resulting in different values for $A_i(\xi,d_i)$ defined as $A_i(\xi,d_i) = \xi \cdot \mathds{1}(d_i = 1)$.

For each $\xi$ we calculate $h_g$ as a function of $A_i(\xi,d_i)$, denoted as $h_g(\xi) = \mathcal{H}(x_{i-h}, \cdots,x_{i-1}, A_i(\xi,d_i))$, and use it to generate a new seed $u_i(\xi) \sim U(0,1,h_g(\xi))$, where \( u_i(\xi) \in \mathbb{R}^V \) is a vector of standard uniform random variables generated with the seed \( h_g(\xi) \). 
For unwatermarked text or when the salt key does not match the generating key, the $x_i$-th random variable in $u_i$ (denoted as $u_{ix_i}$) will simply follow a uniform distribution. However, if the text is watermarked and the salt key is correct, the $x_i$-th random variable in $u_i$ is sampled as the maximum of the Gumbel-modified logits $\ell_{i}$, leading to a Beta distribution \citep{fernandez2023three}.
Similar to the detection method \citep{aaronson2023watermarking,fernandez2023three}, we calculate the score as $S_d = -\sum_i^T (1-d_i) \ln (1-u_{ix_i}(0))$ and $S_k(\xi) = -\sum_i^T d_i \ln (1-u_{ix_i}(\xi))$, where $S_d$ is the indicator score to decide if it is watermarked. Additionally, \(S_k(\xi)\) is the identification watermark score, indicating the likelihood that key \(\xi\) is embedded.
We calculate the detection watermark as $I_d = \mathds{1}(S_d > \tau_d)$ and the user identification watermark $I_k$ as the argument maximizing $S_k(\xi)$ which denotes as $I_k=\arg\max_\xi S_k(\xi)$.

To perform detection, previous works \citep{fernandez2023three,yoo2023advancing} use all tokens to encode the identification watermark and reuse the same statistics to infer the detection watermark. We refer to these methods as Full Key Encoding (FKE), which may be either distribution-based or dictionary-based.
They introduce a threshold $\tau_k$, where samples with the maximum score $\max_\xi S_k(\xi)$ below $\tau_k$ are considered not watermarked. 
We also explore extensions of the DW framework to search for potential improvements. All method and models investigated are shown below:

\textbf{Full Key Encoding (FKE).} We refer all methods \citet{fernandez2023three,yoo2023advancing} that utilize all tokens to encode identification watermark as FKE. FKE utilizes the maximal score over all possible keys and checking if it exceeds a specified threshold $\tau_k$. The sum is calculated as $S_k(\xi)=-\sum_i^T \ln (1-u_{ix_i}(\xi))$ and  \( I_d \) is determined as \( I_d = \mathds{1}(\max_\xi S_k(\xi) > \tau_k) \), where the condition evaluates if the maximum score \( \max_\xi S_k(\xi) \) surpasses \( \tau_k \).

\textbf{Partial Key Encoding (PKE).} To better compare the results of FKE with our proposed strategy, we utilize only a portion of the tokens to encode detection watermark and use this information to determine if the text is watermarked. The remaining tokens are left unused for encoding. In PKE, since only partial tokens are used to encode the identification watermark, the score is calculated as \( S_k'(\xi) = -\sum_{i}^T \mathds{1}{\{A_i(\xi, d_i) \neq 0\}} \ln (1 - u_{ix_i}(\xi)) \). The indicator \( I_d \) is then determined as \( I_d = \mathds{1}(\max_\xi S_k'(\xi) > \tau_k) \). This method serves as an ablation study of the FKE method.

\textbf{Hybrid Dual Watermark (HDW).} This uses both the detection and the identification watermark for detection, reducing detection errors. In HDW, \( I_d \) is calculated as \( I_d = \mathds{1}(S_d > \tau_d \cap \max_\xi S_k(\xi) > \tau_k) \), which evaluates if the score \( S_d \) exceeds a threshold \( \tau_d \) and if the score \( \max_\xi S_k(\xi) \) surpasses the threshold \( \tau_k \) simultaneously.

\textbf{Mean Rebalance (MR).} MR is a natural variant of HDW designed to accommodate variations observed across different sequences. Since the mean value of scores can vary for each sequence, using a fixed threshold may lead to errors. To mitigate this, the MR method compares the maximum score \( S_k(\xi) \) with the mean value of the scores and considers the sequence as unwatermarked if the difference is below a particular threshold. The condition is adjusted as \( I_d = \mathds{1}(S_d > \tau_d \cap \max_\xi S_k(\xi) - \frac{1}{K}\sum_{\xi=1}^K S_k(\xi) > \tau_k) \).

\textbf{Second Rebalance (SR).} Similarly, SR  utilizes the difference between the highest score and the second highest score in the sequence. The indicator $I_d$ is calculated as $I_d=\mathds{1}(S_d > \tau_d \cap \max_\xi S_k(\xi) - \max_{\xi \ne \arg\max S_k(\xi)} S_k(\xi) > \tau_k)$. This condition specifies that the gap between the maximum score and the second largest score must exceed \( \tau_k \).

\subsection{Gumbel-Max Trick Equivalence} \label{sec:relation}
In \citet{aaronson2023watermarking}'s work, they do not explicitly generate the Gumbel variable and select the maximal one. Instead, they perform an equivalent trick by sampling \(V\) random variables \(u = (u_1, \ldots, u_V)\), where \(u_v\) are i.i.d. with \(u_v \sim U(0,1)\). Given the probability vector \(p = (p_1, \ldots, p_V)\), the current token is selected as \(V^{\star} = \arg \max_v u_v^{1/p_v}\). This is known to be equivalent to the Gumbel-Max trick. However, as our analysis depends on this equivalence, we reproduce the derivation in \Cref{thm:equivalent} and \Cref{sec:equivalent}.

\subsection{Adapting Dictionary-Based Methods}\label{sec:adaptdictionary}

Dictionary-based methods \citep{kirchenbauer2023watermark, yoo2023advancing, wangtowards}, such as the multi-bit approach, divide the dictionary into multiple partitions based on the hash of the input information. These methods recover messages by identifying the dominant partition for each bit and detect if a text is watermarked by taking the maximum of multiple binomial statistics \citep{yoo2023advancing}. As we demonstrate in \Cref{sec:multibit}, these approaches suffer from a false detection problem. Our analysis reveals that the statistic is based on the maximal value of multiple binomial variables. Consequently, it no longer represents a binomial distribution but approximates a Gumbel distribution \citep{kotz2000extreme, haan2006extreme}, which inherently preserves the false detection problem.
To adapt our method to the multi-bit backbone, we calculate the hash key $h_i$ and the indicator $d_i$ to determine if the current token encodes an detection watermark or a identification watermark. We then apply the multi-bit algorithm to encode the corresponding information. The remaining procedure aligns with the distribution-based methods.

\section{Theoretical Analysis}\label{sec:theory}
To provide a thorough understanding of the false detection problem inherent in the traditional Full Key Encoding (FKE) method, and to demonstrate how our proposed DW and HDW approaches address this issue, we derive the false positive bound for these methods. Similarly, we first conduct the theoretical analysis for distribution-based methods and defer the analysis of the dictionary-based methods to \Cref{sec:multibit}. We denote capacity of identification watermark as \(K\) (which can be used to represent the total number of user IDs), the generated sequence length as \(T\).
We compute three false positive bounds. The first bound, presented in \Cref{thm:basebound}, provides a theoretical understanding of the baseline FKE method, which considers a document as watermarked if it matches any one of the key ID $\xi$. Next, we conduct the second analysis in \Cref{thm:ourbound} to demonstrate how and why our proposed DW method effectively addresses the false detection problem. Following this, we perform additional analysis in \Cref{thm:hybridbound}, illustrating how the HDW method can further enhance performance. Then, we conduct a numerical experiment based on these bounds  that empirically demonstrates the differences between FKE and our proposed approach DW. We also present a theoretical analysis of multi-bit methods \citep{yoo2023advancing,wang2023towards,kirchenbauer2023watermark} with \Cref{thm:mbit} in \Cref{sec:multibit}, showing that false recognition persists as key capacity increases in dictionary-based methods.

\subsection{False Recognition Bound}\label{sec:bounds}
For FKE, we consider a text sequence to be watermarked if the maximum score $\max_\xi S_k(\xi)$ exceeds a certain threshold $\tau_k$. However, some unwatermarked text may also be mistakenly classified as watermarked. As discussed in \Cref{sec:detecting}, if a text is unwatermarked, all elements in the random variable vector $u_i(\xi)\in \mathbb{R}^V$, drawn during detection follow a uniform distribution. We denote $u_{ix_i}(\xi)$ as the $x_i$-th element of $u_i(\xi)$, which is the random variable corresponding to the generated token $x_i$. We compute the probability that the statistic $S_k(\xi)$ exceeds a specified threshold $\tau_k$, given all $u_{ix_i}(\xi)$ are uniformly distributed. This corresponds to the false positive rate, which lies at the core of the false detection problem. We first establish the following theorem.

\begin{theorem}\label{thm:basebound}

Consider random variables \( u_{ix_i}(\xi) \) drawn from a uniform distribution over \([0, 1]\), where \(\xi = [1, \ldots, K]\) represents the key, and \(K\) denotes the total key capacity. \(u_i \in \mathbb{R}^V \), where \( V \) is the vocabulary size, and \( u_{ix_i} \) corresponds to the \(x_i\)th token in \( u_i \). The index \(i = [1, \ldots, T]\) refers to the \(i\)th token in the generated sequence. The score is calculated as:
\(
S_k(\xi) = -\frac{1}{T}\sum_{i=1}^T \ln (1 - u_{ix_i}(\xi)).
\)
We consider the sample is watermarked if \(\max_\xi S_k(\xi) \ge \tau_k\), where \(\tau_k\) is a threshold parameter chosen to bound the FNR. The false positive probability is bounded as:

\quad \quad \quad \quad \quad \quad 
$
\begin{aligned}\scriptsize
\Pr\left(\max_\xi S_k(\xi) \ge \tau_k\right) \le 1 - \left(1 - \exp \left( T \left( \tau_k \left( \frac{1}{e} - 1 \right) + 1 \right) \right) \right)^K.
\end{aligned}
$

\end{theorem}

For details, please see \Cref{sec:baseproof}. In \Cref{thm:basebound}, we demonstrate the relationship between false recognition and the parameters \(T\), \(K\), and \(\tau_k\) for the traditional FKE method. Our theory shows that as the number of tokens (\(T\)) increases, false recognition rates decrease. However, as the number of keys \(K\) increases, the false positive bound also increases, leading to the aforementioned false detection problem. This implies that with a larger capacity, traditional FKE methods are more likely to mistakenly identify plain text as watermarked text. 

The key benefit of our proposed DW method is that we use a subset of tokens to indicate if the text is watermarked, rather than using all of them to encode identification watermark. To analyze our method, we assume that we use \(T'=\lfloor r_dT \rfloor\) tokens to encode an detection watermark, where \(r_d \in (0,1)\) represents the ratio of such tokens, and \(\lfloor \cdot \rfloor\) denotes the floor function. In \Cref{thm:ourbound}, we prove the upper bound for false recognition in the DW method.

\begin{theorem}\label{thm:ourbound}
Consider random variables \(u_{ix_i}\) drawn from a uniform distribution on \([0, 1]\). The index \(i = [1, \ldots, T']\) represents the \(i^\text{th}\) token of the generated text. We calculate the score
\(
S_d = -\frac{1}{T'}\sum_{i=1}^{T'} \ln (1 - u_{ix_i}).
\)
We regard a sample as watermarked if \(S_d \ge \tau_d\), where \(\tau_d\) is some threshold. Then the false positive probability is bounded by: 

\quad \quad \quad \quad \quad \quad \quad \quad \quad 
$\begin{aligned}\scriptsize
\Pr\left(S_d \ge \tau_d\right) \le \exp \left( T' \left( \tau_d \left( \frac{1}{e} - 1 \right) + 1 \right) \right).
\end{aligned}$

\end{theorem}

For the proof, see \Cref{sec:ourproof}. \Cref{thm:ourbound} demonstrates the relationship between false recognition and the parameters \(T' = \lfloor r_dT \rfloor\) and \(\tau_d\) for the DW method. Similar to FKE, as \(T'\) increases, the false detection problem can be alleviated. i.e. when \(T\) is fixed, increasing \(r_d\) can help mitigate the  false detection problem. It should be noted that the bound is independent of the capacity \(K\), which significantly helps reduce the false detection problem when \(K\) is large. 

\begin{wrapfigure}{r}{0.5\textwidth}
    \centering
    \begin{minipage}[h]{0.5\textwidth}
        \centering
        \includegraphics[width=0.9\linewidth]{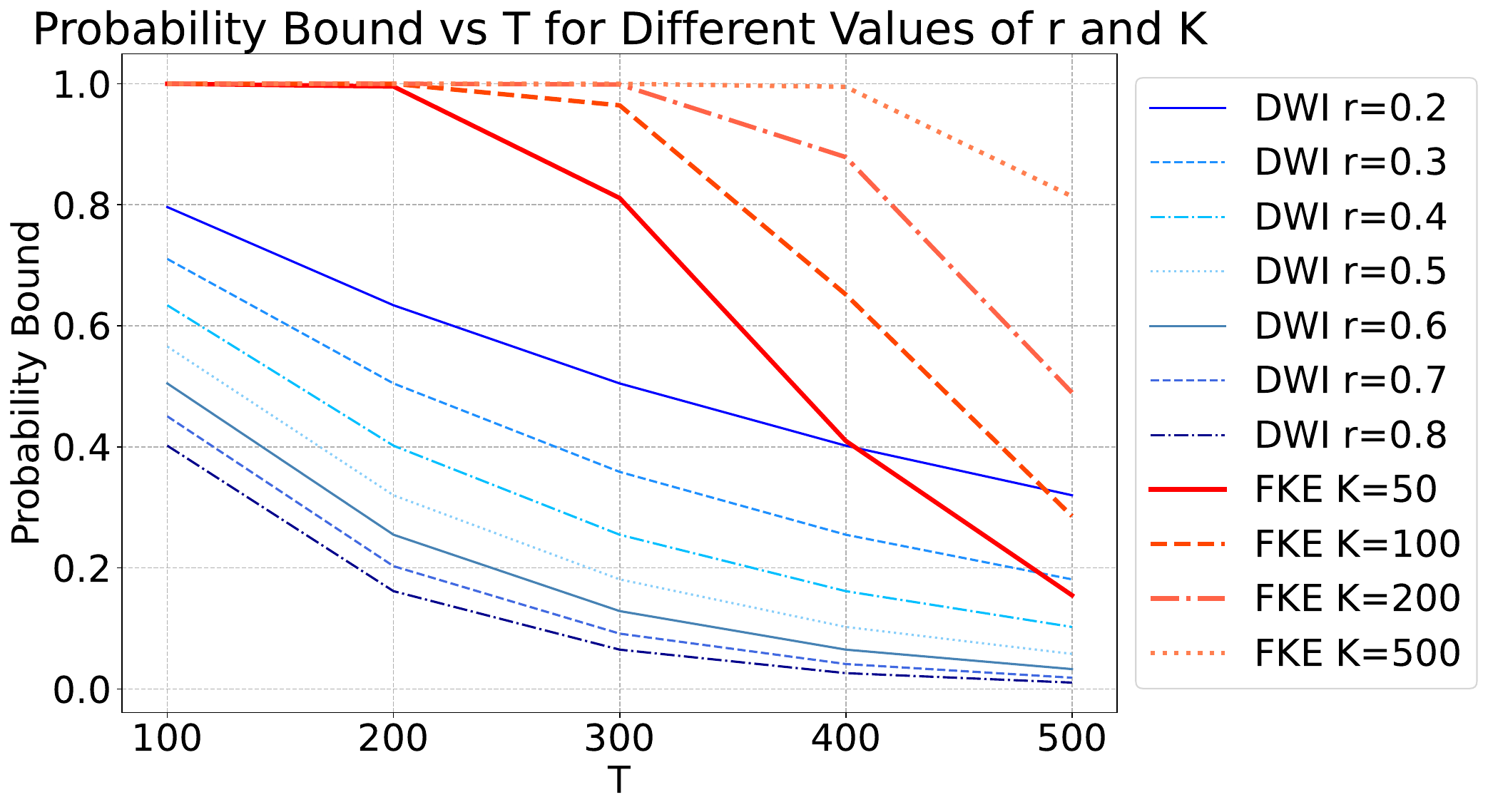}
        \vspace{-1em}
        \caption{Numerical comparison of the probability bounds for DW and FKE methods, presenting the numerical bounds for \Cref{thm:basebound} and \Cref{thm:ourbound}.}
        \label{fig:theory-plot}
    \end{minipage}%
\end{wrapfigure}

To compare FKE and DW methods, we conducted a numerical experiment by plotting the probability bound for \Cref{thm:basebound} with \(K\) ranging from 50 to 500, and the probability bound for \Cref{thm:ourbound} with \(r_d\) ranging from 0.2 to 0.8, as shown in \Cref{fig:theory-plot}. The plot clearly shows that the bound for DW methods is significantly lower than that of FKE, demonstrating the effectiveness of our approach.
As the tokens used in the DW method is smaller than that used in FKE, we calculated the minimal \(K\) value for which DW's bound outperforms FKE's bound, given specific \(T\) and \(r_d\) in \Cref{sec:numerical}.

It is a natural extension to combine the FKE and DW methods to form the HDW method, as discussed in \Cref{sec:detecting}. Here, we also establish a theoretical analysis of the HDW method to demonstrate its effectiveness. The bound is provided in \Cref{thm:hybridbound}.

\begin{theorem}\label{thm:hybridbound}
Using notation introduced in \Cref{thm:basebound} and \Cref{thm:ourbound}, we use \(\lfloor r_dT\rfloor\) tokens to calculate \(S_d\) and \(\lfloor(1-r_d)T\rfloor\) tokens to calculate \(S_k(\xi)\). The hybrid probability \( \Pr(S_d > \tau_d \cap \max_\xi S_k(\xi) > \tau_k) \) is bounded as follows:
\vspace{-1.5em}

$$\begin{aligned}\scriptsize
&\Pr\left(S_d \ge \tau_d \cap \max_\xi S_k(\xi)> \tau_k \right) = \Pr\left(S_d \ge \tau_d \right) \cdot  \Pr\left( \max_\xi S_k(\xi)> \tau_k \right)\nonumber \\ 
&\le \exp \left( \lfloor r_dT \rfloor \left( \tau_d \left( \frac{1}{e} - 1 \right) + 1 \right) \right)\cdot \left(1 - \left(1 - \exp \left( \lfloor (1-r_d)T \rfloor \left( \tau_k \left( \frac{1}{e} - 1 \right) + 1 \right) \right) \right)^K \right)\nonumber
\end{aligned}$$
\vspace{-0.5em}

\end{theorem}

\Cref{thm:hybridbound} is a straightforward combination of the results from \Cref{thm:basebound} and \Cref{thm:ourbound}. A proof can be found in \Cref{sec:hybridproof}. From \Cref{thm:hybridbound}, a similar conclusion can be observed: as \(T\) increases, the bound decreases, indicating an improved ability to alleviate the false detection problem. The effect of \(r_d\) depends on which part is dominant. The bound for the hybrid strategy generally dominates those of method FKE and DW. It strictly dominates strategy DW when \(r_d < 1\) and coincides with it when \(r_d = 1\). Likewise, it coincides with strategy FKE when \(K = 1\) and strictly dominates it when \(K > 1\).

\paragraph{Theoretical Analysis of Dictionary-Based Method}
Dictionary-based methods \citep{kirchenbauer2023watermark, yoo2023advancing, wangtowards} use the binomial approximation for the dominant dictionary partition, but in \Cref{sec:multibit}, we show they follow a Gumbel distribution. As such, the false detection rate grows with key capacity. We also provide numerical experiments to validate it.

\section{Experiments}\label{sec:exps}

\begin{table*}[t]
    \centering
    \scriptsize
    \resizebox{0.99\textwidth}{!}{
    \begin{minipage}[h]{0.28\textwidth}
        \centering
        \scriptsize
        \begin{tabular}{@{~}l@{~}@{~}l@{~}@{~}l@{~}@{~}l@{~}@{~}l@{~}}
        \toprule
         & Accu-I$\uparrow$ & Accu-O$\uparrow$ & FPR$\downarrow$ & Sim$\uparrow$ \\
        \midrule
        FKE & 0.857 & 0.754 & 0.231 & 0.685 \\
        PKE & 0.794 & 0.646 & 0.25 & 0.690 \\
        DW & 0.903 & 0.715 & 0.109 & 0.691 \\
        HDW & 0.906 & 0.718 & 0.092 & 0.690 \\
        MR & 0.907 & 0.719 & 0.0917 & 0.690 \\
        SR & 0.901 & 0.723 & 0.0858 & 0.690 \\
        \bottomrule
        \end{tabular}
        \vspace{1.0em}
        \caption{ Main Experiment. {\scriptsize$\uparrow$}  means higher is better, and  {\scriptsize$\downarrow$} means lower is better.}
        \label{tab:mainexp}
    \end{minipage}%
    \quad \quad 
    \begin{minipage}{0.40\textwidth}
    \begin{minipage}{0.40\textwidth}
        \centering
        \begin{tabular}{@{~}l@{~}@{~}l@{~}@{~}l@{~}@{~}l@{~}}
        \toprule
         & Accu-I$\uparrow$ & Accu-O$\uparrow$ & FPR$\downarrow$ \\
        \midrule
        20 & 0.877 & 0.824 & 0.178 \\
        50 & 0.866 & 0.796 & 0.234 \\
        100 & 0.867 & 0.799 & 0.233 \\
        200 & 0.853 & 0.777 & 0.260 \\
        500 & 0.843 & 0.748 & 0.258 \\
        1000 & 0.832 & 0.735 & 0.289 \\
        2000 & 0.824 & 0.721 & 0.295 \\
        \bottomrule
        \end{tabular}
    \end{minipage}%
    \quad\quad\quad\quad\quad
    \begin{minipage}{0.24\textwidth}
        \centering
        \begin{tabular}{@{~}l@{~}@{~}l@{~}@{~}l@{~}@{~}l@{~}}
        \toprule
         & Accu-I$\uparrow$ & Accu-O$\uparrow$ & FPR$\downarrow$ \\
        \midrule
        20 & 0.887 & 0.786 & 0.173 \\
        50 & 0.904 & 0.786 & 0.149 \\
        100 & 0.901 & 0.777 & 0.12 \\
        200 & 0.904 & 0.757 & 0.127 \\
        500 & 0.902 & 0.742 & 0.103 \\
        1000 & 0.899 & 0.717 & 0.141 \\
        2000 & 0.893 & 0.689 & 0.095 \\
        \bottomrule
        \end{tabular}
    \end{minipage}
    \caption{Key capacity results for FKE (left) and HDW (right) methods at varying key capacities (20 to 2000).}
    \label{tab:keyexp}
    \end{minipage}
    \quad \quad \quad \quad \quad
    \begin{minipage}{0.31\textwidth}
        \centering
        \begin{tabular}{@{~}l@{~}@{~}l@{~}@{~}l@{~}@{~}l@{~}@{~}l@{~}}
        \toprule
         & Accu-I$\uparrow$ & Accu-O$\uparrow$ & FPR$\downarrow$ & Sim$\uparrow$\\
        \midrule
        FKE & 0.926 & 0.883 & 0.120 & 0.550 \\
        PKE & 0.898 & 0.835 & 0.175 & 0.562 \\
        DW & 0.959 & 0.910 & 0.0178 & 0.562 \\
        HDW & 0.955 & 0.909 & 0.0142 & 0.562 \\
        MR & 0.957 & 0.903 & 0.0243 & 0.556 \\
        SR & 0.948 & 0.910 & 0.0174 & 0.562 \\
        \bottomrule
        \end{tabular}
        \vspace{1.1em}
        \caption{Results for dictionary-based methods that utilize the Multi-bit backbone.}
        \label{tab:marylandexp}
    \end{minipage}
    }
    \vspace{-1.5em}
\end{table*}

Our experiments follow the same setup as \citet{fernandez2023three}, using the Guanaco-7b model \citep{dettmers2024qlora}, an instruction fine-tuned LLaMA model \citep{touvron2023llama}, with the first 1,000 prompts from the Alpaca dataset \citep{taori2023stanford}. Our dataset consists of 1,000 samples, including a mix of watermarked and unwatermarked text. A salt key represents IDs ranging from 1 to 1,000. Our framework returns either `None'—indicating that the text is unwatermarked—or an integer representing the key ID. We use 1000 samples mixed with watermarked and unwatermarked text to test all the methods. To assess the models' capability across a range of proportions of watermarked text, we generated 11 datasets with watermarked text ratios ranging from [0\%, 10\%, \ldots, 90\%, 100\%]. We report the overall scores by averaging over watermarked ratios. We use 500 samples as the development set for hyperparameter selection and other 500 as a test set for evaluation with the chosen hyperparameters. 
For hyperparameter selection, we performed a grid search on $\tau_d$ and $\tau_k$, exploring values within the range [0.02, 0.04, \ldots, 7.98, 8.0] for all models on development set. We evaluate using three metrics: Accu-I measures the accuracy of determining if the text is watermarked, irrespective of the correctness of the key prediction. It converts all results to binary outcomes—1 for watermarked and 0 for non-watermarked—and compares these with the gold standard; Accu-O represents the overall accuracy, assessing both the accuracy of watermark indicator predictions and key predictions. and False Positive Ratio (FPR), which indicates the extent of the false detection problem. We follow \citet{fernandez2023three} in using cosine similarity (Sim) between watermarked and unwatermarked text to evaluate generation quality and information loss. A higher cosine similarity score indicates that the generated watermarked text closely resembles the unwatermarked text, reflecting better quality and minimal information loss. We run all models on an NVIDIA A100 GPU with 80GB of memory. 

We compare DW, with the following baseline methods: Full Key Encoding (FKE), Partial Key Encoding (PKE), Hybrid Dual Watermark (HDW), Mean Rebalance (MR), and Second Rebalance (SR). See \Cref{sec:detecting} for details. For the distribution methods, we adopt \citet{fernandez2023three} as backbone. For the dictionary-based method, we utilize Multi-bit \citep{yoo2023advancing, wang2023towards} backbone.

\subsection{Main Experiments} 
We compare DW and HDW methods with several baseline models. From the results shown in \Cref{tab:mainexp}  we see: (1) HDW  and its variants outperform all other approaches in Accu-I and FPR, demonstrating the effectiveness of our proposed methods in detecting watermarked text. The FPR of these models is consistent with our analysis in \Cref{thm:hybridbound}, further showing the correctness of our theory. (2) Compared to the baseline model FKE, DW shows a slight decrease in Accu-O. This  because only half of the tokens are utilized to encode the identification watermark. However, given that this approach significantly mitigates the false detection problem, we found this compromise acceptable. (3) By comparing HDW, SR, and MR, we see that SR and MR effectively improve performance, highlighting the effectiveness of these variant strategies. (4) The performance of PKE lags significantly behind FKE. This discrepancy arises because PKE utilizes only half of the tokens to encode the identification watermark and leaves the remaining tokens unused. This comparison highlights the crucial role of the parameter $T$ in influencing performance and further substantiates the validity of our theory. These results also explain why HDW achieves only marginal improvements, as allocating only half of the tokens to encode the identification watermark can lead to a performance decline. (5) Regarding the FPR score, HDW outperforms both DW and FKE, demonstrating the correctness of our analysis in \Cref{sec:bounds}. (6) The Sim scores across all models are closely clustered around 0.69, suggesting that the quality of the generated text is similar for all models, and that all models maintain an adequate level of similarity to the unwatermarked text.

\subsection{Key Capacity Experiment}\label{sec:keycapacity}

To demonstrate the false detection problem and the effectiveness of our models under different key capacities $K$, we compare the FKE and HDW methods with the capacity $K$ ranging in $[20, 50, 100, 200, 500, 1000, 2000]$. We present the average scores for FKE and HDW in \Cref{tab:keyexp} (left) and \Cref{tab:keyexp} (right), respectively. Additionally, we provide a breakdown of scores relative to different watermarked text ratios in \Cref{sec:keycapacitybreakdown}, respectively. It can be observed from the results that: (1) In the FKE model, since it relies entirely on the maximal score of all keys in the key space, the performance decreases significantly as $K$ increases. This is supported and guaranteed by our analysis in \Cref{thm:basebound}. (2) The HDW results show no significant differences for different key capacities in our HDW method. This is because the indicator is the dominant part and can ensure the FPR avoids the influence of the total key count.

\subsection{Dictionary-Based Method experiments} \label{sec:marylandfpr}

We adapt our DW to Multi-bit model introduced by \citet{yoo2023advancing}, which divides the word dictionary into multiple partitions based on the keys. We report the scores for different models in \Cref{tab:marylandexp} and provide a breakdown of scores relative to varying watermarked text ratios in  \Cref{fig:marylandfpr}. The results indicate that (1) The results of the dictionary-based model are similar to those of the distribution-based model. These findings demonstrate that our proposed method is generalizable to other watermarking approaches and effectively mitigates the false detection problem. (2) Across different models, our methods consistently outperform the FKE model significantly, even with the Multi-bit backbone. This demonstrates that our approach is well-suited for dictionary-based models. 


\subsection{Additional Experiments}

We conduct more extensive experiments to show the effectiveness and limitations of our models in various settings. Due to the page limit, we defer the details to the appendix.

In \Cref{sec:watermarktextratio}, we analyze the relationship between watermarked text ratios and metrics such as Accu-I, Accu-O, and FPR, highlighting performance trends under different watermark ratios. In \Cref{watermarktextratiodictionary}, we extend this analysis to the dictionary-based method, showcasing the adaptability of our models across different backbones.
In \Cref{sec:insdel}, we evaluate robustness against insertion and deletion attacks, demonstrating that our models maintain comparable performance to baseline methods despite token-level modifications. Similarly, in \Cref{sec:paraphraseattack}, we assess resilience to paraphrase attacks, showing that even with text rephrasing, our models effectively recognize watermarked content.
In \Cref{sec:runtimeanalysis}, we perform a runtime analysis to compare computational costs of watermarking and detection processes across various models, revealing that our approach is efficient and scalable. In \Cref{sec:seqlenanalysis}, we explore the impact of sequence length on performance, finding that longer sequences enhance accuracy while introducing slight increases in FPR.
To validate the generalizability of our models, we evaluate performance on domain-specific datasets in \Cref{sec:moredata}, using BioASQ and LegalQA. The results confirm robust performance in specialized contexts with minimal alteration to text similarity.
In \Cref{sec:indicationratioparam}, we study the effect of the indication ratio parameter $r_d$, demonstrating how varying token allocations for watermark detection impacts accuracy and FPR. 
In \Cref{sec:indicationratio}, we further investigate how different ratios of $r_d$ impact FPR, showing that employing more tokens to encode the watermark indicator effectively mitigates FPR. This analysis also highlights the effect of watermarked text ratios on threshold tuning and resulting FPR trends.
In \Cref{sec:samplesizeexp}, we evaluate the sensitivity of our models to variations in sample size, confirming that our results are robust and not significantly influenced by sample size. 
In \Cref{sec:winsizeparam}, we examine the influence of the window size parameter $h$, finding that text quality remains stable across a wide range of values.
In \Cref{sec:keycapacitybreakdown}, we conduct a key capacity breakdown analysis to investigate how different watermark ratios affect detection performance. \Cref{sec:keycapacitybreakdowndictionary} extends this analysis to the dictionary-based method, emphasizing the flexibility of our models in diverse settings.
Lastly, in \Cref{sec:keycapacityexpfordicbase}, we conduct key capacity experiments for dictionary-based methods using the Multi-bit as backbone. The results highlight that while traditional dictionary-based methods show an increase in FPR as key capacity grows, our HDW method effectively maintains a consistent FPR scale.
These experiments provide deeper insights into the capabilities and limitations of our proposed methods.

\section{Conclusion}
In this paper, we address the false detection problem in watermarking methods for text generated by LLMs. We establish a theoretical bound demonstrating the inherent inevitability of false positive errors in watermarking techniques like FKE. To mitigate this problem, we introduce a novel DW method that jointly encodes indicator and identification watermark. Furthermore, we present a analysis of our proposed method and validate it through extensive empirical experiments. Our results, both theoretical and empirical, indicate that the DW method and its variants effectively reduces the false positive ratio, thereby alleviating the false detection problem. This enhancement in watermarking reliability can significantly promote the trustworthiness of LLM-generated content.

\bibliographystyle{unsrtnat}
\bibliography{references}

\newpage
\appendix
\onecolumn

\clearpage
\appendix
\onecolumn
\begin{center}
  {\LARGE \textbf{Appendix. Supplementary Material}}
\end{center}
\renewcommand{\thesection}{A.\arabic{section}}
\setcounter{theorem}{0}
\setcounter{equation}{0}

\section{Impact Statements}\label{sec:impact}

This work proposes the Dual Watermark (DW) scheme to address the false detection problem in watermarking LLM outputs. By significantly reducing false positives while maintaining high detection accuracy, our method enables reliable identification of watermarked text and user-specific tracking, even at large key capacities. These advancements improve the scalability and robustness of watermarking techniques, making them more effective for managing LLM-generated content across diverse applications.

The contributions of this work have several positive implications, such as enhancing content authenticity, improving accountability, and enabling more reliable tracing of LLM outputs. However, the adoption of such techniques is not without potential risks. For instance, the identification-based watermarking introduced in our approach also raises privacy concerns if misused to track or de-anonymize users without their consent. Additionally, improper disclosure or misuse of these tools may lead to unintended consequences, such as penalizing legitimate users based on incorrect or biased watermark detection.

Therefore, while this work provides critical advancements in watermarking techniques, we urge careful governance and transparent usage policies to mitigate potential harms and ensure that these tools are deployed responsibly.

\section{Notation} 
We denote the key ID by \(\xi\in [1,K]\), with \(K\in \mathbb{N}\) representing the total number of keys i.e. the capacity. The token sequence \([x_1, \ldots, x_T]\) is generated by an LLM \(\mathcal{L}\), where each token \(x_i\) is within the range \([1, \ldots, V]\), \(V\) being the vocabulary size. The probabilities for predicting the $i$th token are denoted as \(p_i\), with corresponding logits \(\ell_i\), which are adjusted to \(\tilde{\ell}_i\) by incorporating the Gumbel variable $g_i$. The indicator function \(\mathds{1}(\cdot)\) returns 1 if the condition is true and 0 otherwise. The hash function \(\mathcal{H}\) is used to calculate the hash key. A uniform distribution \(U(0, 1, s)\) is used, where \(s\) is the seed used to generate the standard uniform random variables. The scores \(S_d\) and \(S_k(\xi)\) are statistics used to detect the detection watermark and the identification watermark for key $\xi$, respectively. The parameter $r_d$ determines the proportion of tokens used to encode the detection watermark, while $\tau_d$ and $\tau_k$ serve as threshold variables for detecting the detection watermark and the identification watermark, respectively.

\section{Proof of \Cref{thm:basebound}}\label{sec:baseproof}
\begin{theorem}

Consider random variables \( u_{ix_i}(\xi) \) drawn from a uniform distribution over \([0, 1]\), where \(\xi = [1, \ldots, K]\) represents the key, and \(K\) denotes the total key capacity. \(u_i \in \mathbb{R}^V \), where \( V \) is the vocabulary size, and \( u_{ix_i} \) corresponds to the \(x_i\)th token in \( u_i \). The index \(i = [1, \ldots, T]\) refers to the \(i\)th token in the generated sequence. The score is calculated as:
\(
S_k(\xi) = -\frac{1}{T}\sum_{i=1}^T \ln (1 - u_{ix_i}(\xi)).
\)
We consider the sample is watermarked if \(\max_\xi S_k(\xi) \ge \tau_k\), where \(\tau_k\) is a threshold parameter chosen to bound the FNR. The false positive probability is bounded as:
\begin{equation}\scriptsize
\Pr\left(\max_\xi S_k(\xi) \ge \tau_k\right) \le 1 - \left(1 - \exp \left( T \left( \tau_k \left( \frac{1}{e} - 1 \right) + 1 \right) \right) \right)^K.
\end{equation}
\end{theorem}

\begin{proof}[Proof Sketch]
From \Cref{lemma:expuniform}, we know that if \(r\) is uniformly distributed over \([0, 1]\), then \(X = -\ln(1 - r)\) follows an exponential distribution with parameter 1. 
According to \Cref{lemma:Sisgamma}, if \(u_{ix_i}(\xi)\) are independent and uniformly distributed over \([0, 1]\), then \(S_k(\xi) = -\frac{1}{T}\sum_{i=1}^T \ln(1 - u_{ix_i}(\xi)) \sim \text{Gamma}(T, \frac{1}{T})\). 
Using \Cref{lemma:gammabound}, given \(X \sim \text{Gamma}(T, \frac{1}{T})\), with probability \(1 - \delta\), \(X \leq \frac{\frac{\log \delta}{T} - 1}{1/e - 1}\).
Therefore, for \(S_k(\xi) \sim \text{Gamma}(T, \frac{1}{T})\), with probability \(1 - \delta\), \(S_k(\xi) \leq \frac{\frac{\log \delta}{T} - 1}{1/e - 1}\). Given this bound for each \(S_k(\xi)\), we use \Cref{lemma:cnbound} to bound the probability of the maximum \(S_k(\xi)\) over $\xi$.
Specifically, \Cref{lemma:cnbound} states that \(\Pr\left(\max_\xi S_k(\xi) \le \tau_k \right) \ge \left( 1 - \exp \left( T \left( \tau_k \left( \frac{1}{e} - 1 \right) + 1 \right) \right) \right)^K\). Taking the complement, we get \(\Pr\left(\max_\xi S_k(\xi) \ge \tau_k\right) \le 1 - \left( 1 - \exp \left( T \left( \tau_k \left( \frac{1}{e} - 1 \right) + 1 \right) \right) \right)^K\). 
This completes the proof of the theorem. 
\end{proof}

\begin{proof}

We will prove the theorem step-by-step using the provided lemmas.

\textbf{Step 1: Showing the Transformation is Exponential}

From \Cref{lemma:expuniform}, we know that if \(r\) is uniformly distributed over \([0, 1]\), then \(X = -\ln(1 - r)\) follows an exponential distribution with parameter 1, i.e., \(X \sim \text{Exp}(1)\).

\textbf{Step 2: Distribution of \(S_k(\xi)\)}

From \Cref{lemma:Sisgamma}, we know that if \(u_{ix_i}(\xi)\) are independent and uniformly distributed over \([0, 1]\), then
\[
S_k(\xi) = -\frac{1}{T}\sum_{i=1}^T \ln(1 - u_{ix_i}(\xi)) \sim \text{Gamma}\left(T, \frac{1}{T}\right).
\]

\textbf{Step 3: Bounding \(S_k(\xi)\)}

Using \Cref{lemma:gammabound}, given \(X \sim \text{Gamma}(T, \frac{1}{T})\), with probability \(1 - \delta\),
\[
X \leq \frac{\frac{\log \delta}{T} - 1}{1/e - 1}.
\]

Therefore, for \(S_k(\xi) \sim \text{Gamma}(T, \frac{1}{T})\), with probability \(1 - \delta\),
\[
S_k(\xi) \leq \frac{\frac{\log \delta}{T} - 1}{1/e - 1}.
\]

\textbf{Step 4: Probability Bound on Maximum \(S_k(\xi)\)}

Given \(S_k(\xi) \leq \frac{\frac{\log \delta}{T} - 1}{1/e - 1}\) with probability \(1 - \delta\) for each \(i\), we use \Cref{lemma:cnbound} to bound the probability of the maximum \(S_k(\xi)\).

From \Cref{lemma:cnbound}, we have:
\[
\Pr\left(\max_\xi S_k(\xi) \le \tau_k\right) \ge \left( 1 - \exp \left( T \left( \tau_k \left( \frac{1}{e} - 1 \right) + 1 \right) \right) \right)^K.
\]

\textbf{Step 5: Complement of Maximum Bound}

To find the probability that the maximum of \(S_k(\xi)\) exceeds \(\tau_k\), we take the complement of the bound derived above:
\[
\Pr\left(\max_\xi S_k(\xi) \ge \tau_k\right) \le 1 - \left( 1 - \exp \left( T \left( \tau_k \left( \frac{1}{e} - 1 \right) + 1 \right) \right) \right)^K.
\]

\textbf{Conclusion}

We have shown that the probability of the maximum score \(\max_\xi S_k(\xi)\) being greater than or equal to \(\tau_k\) is bounded by
\[
\Pr\left(\max_\xi S_k(\xi) \ge \tau_k\right) \le 1 - \left( 1 - \exp \left( T \left( \tau_k \left( \frac{1}{e} - 1 \right) + 1 \right) \right) \right)^K.
\]

This completes the proof of the theorem. With probability \(1 - \delta\),
\[
\Pr\left(\max_\xi S_k(\xi) \ge \tau_k\right) \le 1 - \left( 1 - \exp \left( T \left( \tau_k \left( \frac{1}{e} - 1 \right) + 1 \right) \right) \right)^K.
\]

\end{proof}

\section{Proof of \Cref{thm:ourbound}}\label{sec:ourproof}

\begin{theorem}
Consider random variables \(u_{ix_i}\) drawn from a uniform distribution on \([0, 1]\). The index \(i = [1, \ldots, T']\) represents the \(i\)th token of the generated text. We calculate the score as
\(
S_d = -\frac{1}{T'}\sum_{i=1}^{T'} \ln (1 - u_{ix_i}).
\)
We regard the sample as watermarked if \(S_d \ge \tau_d\), where \(\tau_d\) is a threshold parameter. Then the false positive probability is bounded as follows: 
\begin{equation}
\Pr\left(S_d \ge \tau_d\right) \le \exp \left( T' \left( \tau_d \left( \frac{1}{e} - 1 \right) + 1 \right) \right).
\end{equation}
\end{theorem}

\begin{proof}[Proof Sketch]
Given \(u_{ix_i}\) drawn from a uniform distribution on \([0, 1]\), we transform \(u_{ix_i}\) using \( -\ln(1 - u_{ix_i}) \), which follows an exponential distribution with parameter 1, as shown in \Cref{lemma:expuniform}. The score \( S_d \) is then the average of \( T' \) such transformed variables, scaled by \(-\frac{1}{T'}\), which, by \Cref{lemma:Sisgamma}, follows a Gamma distribution with shape parameter \( T' \) and scale parameter \(\frac{1}{T'}\). Using the bound from \Cref{lemma:gammabound}, with probability \(1 - \delta\), \( S_d \) is less than or equal to a certain function of \(\log \delta\). By expressing \(\delta\) in terms of \( \tau_d \) and solving, we derive that the false positive probability \( \Pr(S_d \ge \tau_d) \) is bounded by an exponential function \(\exp \left( T' \left( \tau_d \left( \frac{1}{e} - 1 \right) + 1 \right) \right)\). Thus, the false positive probability is bounded as claimed in the theorem. 
\end{proof}

\begin{proof}
We will prove this theorem using the following lemmas.

From \Cref{lemma:expuniform}, we know that if \(u_{ix_i} \sim \text{Uniform}(0, 1)\), then \(-\ln(1 - u_{ix_i}) \sim \text{Exp}(1)\).

Using \Cref{lemma:Sisgamma}, the sum of \(T'\) independent exponential random variables follows a Gamma distribution with shape parameter \(T'\) and scale parameter \(\frac{1}{T'}\):
\[
S_d = -\frac{1}{T'} \sum_{i=1}^{T'} \ln(1 - u_{ix_i}) \sim \text{Gamma}(T', \frac{1}{T'}).
\]

From \Cref{lemma:gammabound}, we know that for \(X \sim \text{Gamma}(T', \frac{1}{T'})\), with shape parameter \(T'\) and scale parameter \(\frac{1}{T'}\), with probability \(1 - \delta\):
\[
X \leq \frac{\frac{\log \delta}{T'} - 1}{1/e - 1}.
\]

Adapting this for our case where the shape parameter is \(T'\), with probability \(1 - \delta\):
\[
S_d \leq \frac{\frac{\log \delta}{T'} - 1}{1/e - 1}.
\]

Let’s express \(\delta\) as a function of \(\tau_d\). We set:
\[
\tau_d = \frac{\frac{\log \delta}{T'} - 1}{1/e - 1}.
\]

Solving for \(\log \delta\):
\[
\tau_d (1/e - 1) = \frac{\log \delta}{T'} - 1,
\]
\[
\log \delta = T' (\tau_d (1/e - 1) + 1).
\]

Thus, we have:
\[
\delta = \exp \left( T' (\tau_d (1/e - 1) + 1) \right).
\]

The false positive probability \(\Pr(S_d \ge \tau_k)\) is given by \(\delta\):
\[
\Pr(S_d \ge \tau_k) = \exp \left( T' (\tau_d (1/e - 1) + 1) \right).
\]

Thus, we have shown that the false positive probability is bounded as follows: with probability \(1 - \delta\),
\[
\Pr\left(S_d \ge \tau_k\right) \le \exp \left( T' \left( \tau_d \left( \frac{1}{e} - 1 \right) + 1 \right) \right).
\]
\end{proof}

\section{Proof of \Cref{thm:hybridbound}}\label{sec:hybridproof}

\begin{theorem}
Using notation introduced in \Cref{thm:basebound} and \Cref{thm:ourbound}, we use \(\lfloor r_dT\rfloor\) tokens to calculate \(S_d\) and \(\lfloor(1-r_d)T\rfloor\) tokens to calculate \(S_k(\xi)\). The hybrid probability \( \Pr(S_d > \tau_d \cap \max_\xi S_k(\xi) > \tau_k) \) is bounded as follows:
\begin{align}
\Pr&\left(S_d \ge \tau_d \cap \max_\xi S_k(\xi)> \tau_k \right) = \Pr\left(S_d \ge \tau_d \right) \cdot  \Pr\left( \max_\xi S_k(\xi)> \tau_k \right)\\ 
&\le \exp \left( \lfloor r_dT \rfloor \left( \tau_d \left( \frac{1}{e} - 1 \right) + 1 \right) \right) \left(1 - \left(1 - \exp \left( \lfloor (1-r_d)T \rfloor \left( \tau_k \left( \frac{1}{e} - 1 \right) + 1 \right) \right) \right)^K \right)
\end{align}
\end{theorem}

\begin{proof}
We start by considering the two probabilities involved in the hybrid probability \( \Pr(S_d \ge \tau_d \cap \max_\xi S_k(\xi) > \tau_k) \). By the definition of joint probability for independent events, we can express the hybrid probability as the product of the individual probabilities:

\[
\Pr\left(S_d \ge \tau_d \cap \max_\xi S_k(\xi) > \tau_k \right) = \Pr\left(S_d \ge \tau_d \right) \cdot \Pr\left(\max_\xi S_k(\xi) > \tau_k \right).
\]

Given that \(S_d\) is calculated using \(\lfloor r_dT \rfloor\) tokens and \(S_k(\xi)\) is calculated using \(\lfloor(1-r_d)T \rfloor\) tokens, we can apply the bounds from \Cref{thm:ourbound} and \Cref{thm:basebound} respectively.

First, by applying the bound from \Cref{thm:ourbound} to the probability \(\Pr(S_d \ge \tau_d)\), we have:

\[
\Pr\left(S_d \ge \tau_d\right) \le \exp \left( \lfloor r_dT \rfloor \left( \tau_d \left( \frac{1}{e} - 1 \right) + 1 \right) \right).
\]

Next, by applying the bound from \Cref{thm:basebound} to the probability \(\Pr(\max_\xi S_k(\xi) > \tau_k)\), we obtain:

\[
\Pr\left(\max_\xi S_k(\xi) > \tau_k \right) \le 1 - \left(1 - \exp \left( \lfloor(1-r_d)T\rfloor \left( \tau_k \left( \frac{1}{e} - 1 \right) + 1 \right) \right) \right)^K.
\]

Thus, combining these two results, the hybrid probability can be bounded as follows:

\begin{align}
\Pr\left(S_d \ge \tau_d \cap \max_\xi S_k(\xi) > \tau_k \right) \le &\exp \left( \lfloor r_dT \rfloor \left( \tau_d \left( \frac{1}{e} - 1 \right) + 1 \right) \right) \cdot \\&\left(1 - \left(1 - \exp \left( \lfloor(1-r_d)T\rfloor \left( \tau_k \left( \frac{1}{e} - 1 \right) + 1 \right) \right) \right)^K \right).
\end{align}

This completes the proof.
\end{proof}

\section{Proof of the Equivalence of Gumbel-Max Trick}\label{sec:equivalent}

\begin{proposition}\label{thm:equivalent}
Consider a discrete distribution \(p = (p_1, \ldots, p_V)\) and \(V\) random variables \(u = (u_1, \ldots, u_V)\) such that \(u_v\) are i.i.d. with \(u_v \sim \mathcal{U}_{[0,1]}\). Let \(V^{\star} = \arg \max_v u_v^{1/p_v}\). Define \(G_v = \log(p_v) + g_v\), where \(g_v = -\log(-\log(u_v))\). Then
\[ V^{\star} = G^{\star} \]
\end{proposition}

\begin{proof}
$$
\begin{aligned}
    \arg \max_v G_v &= \arg \max_v \left(\log(p_v) + g_v\right) \\
    &= \arg \max_v \left(\log(p_v) - \log(-\log(u_v))\right) \\
    &= \arg \max_v \exp \left(\log(p_v) - \log(-\log(u_v))\right) \\
    &= \arg \max_v \left(\exp(\log(p_v)) \cdot \exp(-\log(-\log(u_v)))\right) \\
    &= \arg \max_v \left(p_v \cdot \frac{1}{-\log(u_v)}\right) \\
    &= \arg \min_v \left(-\frac{\log(u_v)}{p_v}\right) \\
    &= \arg \max_v \left(\frac{\log(u_v)}{p_v}\right) \\
    &= \arg \max_v \left(\log(u_v^{1/p_v})\right) \\
    &= \arg \max_v \left(u_v^{1/p_v}\right)
    \end{aligned}
$$

Therefore,
\[ 
V^{\star} = \arg \max_v u_v^{1/p_v} 
\]

Thus, the theorem is proved:
\[ 
V^{\star} = G^{\star} 
\]
\end{proof}

\section{Lemmas}

\begin{lemma}\label{lemma:expuniform}
Let $r$ be a random variable uniformly distributed over the interval $[0, 1]$. Define $X = -\ln(1 - r)$. Then $X$ follows an exponential distribution with parameter 1, i.e., $X \sim \text{Exp}(1)$.
\end{lemma}

\begin{proof}
To show that $X = -\ln(1 - r)$ follows an exponential distribution with parameter 1, we first find the cumulative distribution function (CDF) of $X$.

For any $x \geq 0$,
\begin{align*}
F_X(x) &= P(X \leq x) \\
&= P(-\ln(1 - r) \leq x) \\
&= P(\ln(1 - r) \geq -x) \\
&= P(1 - r \geq e^{-x}) \\
&= P(r \leq 1 - e^{-x}).
\end{align*}

Since $r$ is uniformly distributed over $[0, 1]$, its CDF is $F_r(r) = r$. Therefore,
\[
F_X(x) = 1 - e^{-x}, \quad \text{for } x \geq 0.
\]

Next, we differentiate the CDF to obtain the probability density function (PDF):
\[
f_X(x) = \frac{d}{dx} F_X(x) = \frac{d}{dx} (1 - e^{-x}) = e^{-x}, \quad \text{for } x \geq 0.
\]

The PDF $f_X(x) = e^{-x}$ is the PDF of an exponential distribution with parameter 1. Therefore, $X \sim \text{Exp}(1)$.
\end{proof}

\begin{lemma}\label{lemma:Sisgamma}
Let $r_i$ be independent and uniformly distributed over the interval $[0, 1]$ for $i = 1, 2, \ldots, T$. Define $S = -\frac{1}{T} \sum_{i=1}^{T} \ln(1 - r_i)$. Then $S$ follows a Gamma distribution with shape parameter $T$ and scale parameter $\frac{1}{T}$, i.e., $S \sim \text{Gamma}(T, \frac{1}{T})$.
\end{lemma}

\begin{proof}
From \Cref{lemma:expuniform}, we know that if $r$ is uniformly distributed over $[0, 1]$, then $X = -\ln(1 - r) \sim \text{Exp}(1)$.

Given that $r_i \sim \text{Uniform}(0, 1)$, it follows that $u_i = -\ln(1 - r_i) \sim \text{Exp}(1)$ for each $i$.

Now, consider the sum of $T$ such independent exponential random variables:
\[
Y = \sum_{i=1}^{T} u_i
\]

Since the sum of $T$ independent $\text{Exp}(1)$ random variables follows a Gamma distribution with shape parameter $T$ and scale parameter $1$, we have:
\[
Y \sim \text{Gamma}(T, 1)
\]

Next, consider the scaled variable:
\[
S = \frac{Y}{T}
\]

Since $Y \sim \text{Gamma}(T, 1)$, scaling $Y$ by $1/T$ (which is equivalent to dividing by $T$) gives us a new Gamma distributed random variable with the same shape parameter $T$ and a scale parameter of $1/T$. Therefore:
\[
S \sim \text{Gamma}\left(T, \frac{1}{T}\right)
\]

Thus, we have shown that $S = -\frac{1}{T} \sum_{i=1}^{T} \ln(1 - r_i)$ follows a Gamma distribution with shape parameter $T$ and scale parameter $\frac{1}{T}$.
\end{proof}

\begin{lemma}\label{lemma:gammabound}
Given \(X \sim \text{Gamma}(T, \frac{1}{T})\), with shape parameter \(T\) and scale parameter \(\frac{1}{T}\), we can state:

With probability \(1 - \delta\),

\[
X \leq \frac{\frac{\log \delta}{T}-1}{1/e-1}.
\]
\end{lemma}

\begin{proof}
We use the Chernoff bound to derive this result.

First, recall the moment generating function (MGF) of \(X \sim \text{Gamma}(T, \frac{1}{T})\):

\[
M_X(t) = \mathbb{E}[e^{tX}] = \left(1 - \frac{t}{T}\right)^{-T},
\]

for \( t < T \).

Using the Chernoff bound, for any \(t > 0\), we have:

\[
\mathbb{P}(X \geq a) = \mathbb{P}(e^{tX} \geq e^{ta}) \leq \frac{\mathbb{E}[e^{tX}]}{e^{ta}} = \frac{M_X(t)}{e^{ta}}.
\]

Substituting the MGF, we get:

\[
\mathbb{P}(X \geq a) \leq \frac{\left(1 - \frac{t}{T}\right)^{-T}}{e^{ta}}.
\]

To optimize this bound, we need to minimize the right-hand side with respect to \(t\). Therefore, we have:

\[
\log \left( \frac{\left(1 - \frac{t}{T}\right)^{-T}}{e^{ta}} \right) = -T \log\left(1 - \frac{t}{T}\right) - ta.
\]

Differentiate with respect to \(t\) and set the derivative to zero to find the optimal \(t\):

\[
\begin{aligned}
    \frac{d}{dt} \left( -T \log\left(1 - \frac{t}{T}\right) - ta \right) &= 0 \\
    -T \cdot \left( -\frac{1}{T} \cdot \frac{1}{1 - \frac{t}{T}} \right) - a &= 0 \\
    \frac{1}{1 - \frac{t}{T}} - a &= 0 \\
    1 - \frac{1}{a} &= \frac{t}{T} \\
    t &= T \left( 1 - \frac{1}{a} \right).
\end{aligned}
\]

Substituting \(t = T\left(1 - \frac{1}{a}\right)\) back into the Chernoff bound, we have:

\[
\mathbb{P}(X \geq a) \leq \exp \left( -T \log\left(1 - \left(1 - \frac{1}{a}\right)\right) - T\left(1 - \frac{1}{a}\right) a \right).
\]

Simplifying further:

\[
\mathbb{P}(X \geq a) \leq \exp \left( -T \log\left(\frac{1}{a}\right) - T \left(a - 1\right) \right).
\]

For \(a > 0\), we can simplify the expression:

\[
\begin{aligned}
\mathbb{P}(X \geq a) &\leq \exp \left( T \log(a) - T(a - 1) \right)\\
&\leq \exp \left( T \frac{a}{e} - T(a - 1) \right)\\
&= \exp \left( T (\frac{a}{e} - a + 1) \right)
\end{aligned}
\]

Setting this bound to \(\delta\), we get:

\[
\exp \left( T (\frac{a}{e} - a + 1) \right) = \delta.
\]

Taking the natural logarithm:

\[
T (\frac{a}{e} - a + 1) = \log \delta,
\]

\[
a = \frac{\frac{\log \delta}{T}-1}{1/e-1},
\]

Therefore, with probability \(1 - \delta\):

\[
X \leq \frac{\frac{\log \delta}{T}-1}{1/e-1}.
\]

\end{proof}

\begin{lemma}\label{lemma:cnbound}
Given random variables \(u_1, u_2, \ldots, u_K\) where \(K > 0\), such that with probability \(1 - \delta\):
\[
u_i \leq \frac{\frac{\log \delta}{T} - 1}{1/e - 1},
\]
it follows that:
\[
\Pr\left(\max_i u_i \le s\right) \ge \left( 1 - \exp \left( T \left( s \left( \frac{1}{e} - 1 \right) + 1 \right) \right) \right)^K.
\]
\end{lemma}

\begin{proof}
Given the condition:
\[
\Pr\left(u_i \leq \frac{\frac{\log \delta}{T} - 1}{1/e - 1}\right) \ge 1 - \delta,
\]
we denote:
\[
b = \frac{\frac{\log \delta}{T} - 1}{1/e - 1}.
\]

We aim to express this condition in terms of \(s\) and derive a bound for:
\[
\Pr\left(\max_i u_i \le s\right).
\]

First, consider:
\[
\Pr\left(u_i \leq b\right) \ge 1 - \delta.
\]

We need to find a function of \(s\) that relates \(\delta\) to \(s\). Suppose \(s \ge b\). Then:
\[
\Pr\left(u_i \leq s\right) \ge \Pr\left(u_i \leq b\right) \ge 1 - \delta.
\]

We aim to find the probability that all \(u_i\) are less than or equal to \(s\):
\[
\Pr\left(\max_i u_i \le s\right) = \Pr\left(u_1 \le s, u_2 \le s, \ldots, u_K \le s\right).
\]

Assuming the \(u_i\) are independent, we can write:
\[
\Pr\left(u_1 \le s, u_2 \le s, \ldots, u_K \le s\right) = \prod_{i=1}^K \Pr\left(u_i \le s\right).
\]

Since:
\[
\Pr\left(u_i \le s\right) \ge 1 - \delta,
\]
we have:
\[
\Pr\left(\max_i u_i \le s\right) \ge (1 - \delta)^K.
\]

Now, we need to express \(\delta\) in terms of \(s\). Recall the expression for \(b\):
\[
b = \frac{\frac{\log \delta}{T} - 1}{1/e - 1}.
\]

Solving for \(\log \delta\), we get:
\[
b (1/e - 1) = \frac{\log \delta}{T} - 1,
\]
\[
b (1/e - 1) + 1 = \frac{\log \delta}{T},
\]
\[
T \left( b (1/e - 1) + 1 \right) = \log \delta,
\]
\[
\delta = \exp \left( T \left( b (1/e - 1) + 1 \right) \right).
\]

Now, substitute \(b = s\):
\[
\delta = \exp \left( T \left( s \left( \frac{1}{e} - 1 \right) + 1 \right) \right).
\]

Hence:
\[
\Pr\left(\max_i u_i \le s\right) \ge \left( 1 - \exp \left( T \left( s \left( \frac{1}{e} - 1 \right) + 1 \right) \right) \right)^K.
\]

This completes the proof of the lemma.
\end{proof}

\begin{lemma}\label{lemma:xbound}
Given a random variable \(X\), such that with probability \(1 - \delta\):
\[
X \leq \frac{\frac{\log \delta}{r_dT} - 1}{1/e - 1},
\]
it follows that:
\[
\Pr\left(X \le s\right) \ge 1 - \exp \left( r_dT \left( s \left( \frac{1}{e} - 1 \right) + 1 \right) \right).
\]
\end{lemma}

\begin{proof}
Given the condition:
\[
\Pr\left(X \leq \frac{\frac{\log \delta}{r_dT} - 1}{1/e - 1}\right) \ge 1 - \delta,
\]
we denote:
\[
b = \frac{\frac{\log \delta}{r_dT} - 1}{1/e - 1}.
\]

We aim to express \(\delta\) as a function of \(s\) and find the probability bound for \(X \le s\).

Rearranging the expression for \(b\):
\[
b = \frac{\frac{\log \delta}{r_dT} - 1}{1/e - 1},
\]
we solve for \(\log \delta\):
\[
b \left(\frac{1}{e} - 1\right) = \frac{\log \delta}{r_dT} - 1,
\]
\[
b \left(\frac{1}{e} - 1\right) + 1 = \frac{\log \delta}{r_dT},
\]
\[
r_dT \left(b \left(\frac{1}{e} - 1\right) + 1\right) = \log \delta,
\]
\[
\delta = \exp \left( r_dT \left( b \left(\frac{1}{e} - 1\right) + 1 \right) \right).
\]

Next, we relate \(b\) to \(s\). Suppose \(s \ge b\), then:
\[
\Pr\left(X \le s\right) \ge \Pr\left(X \le b\right) \ge 1 - \delta.
\]

Substitute \(b\) with \(s\):
\[
b = s.
\]

Now we have:
\[
\delta = \exp \left( r_dT \left( s \left(\frac{1}{e} - 1\right) + 1 \right) \right).
\]

Thus:
\[
\Pr\left(X \le s\right) \ge 1 - \delta,
\]
where \(\delta = \exp \left( r_dT \left( s \left(\frac{1}{e} - 1\right) + 1 \right) \right)\).

Therefore:
\[
\Pr\left(X \le s\right) \ge 1 - \exp \left( r_dT \left( s \left(\frac{1}{e} - 1 \right) + 1 \right) \right).
\]

This completes the proof of the lemma.
\end{proof}

\begin{lemma}\label{lemma:Kbound}
When $ K \ge \frac{\ln(1 - \exp(r_dT (s (\frac{1}{e} - 1) + 1)))}{\ln(1 - \exp(T (s (\frac{1}{e} - 1) + 1)))} $, it follows that:
\[
1 - \exp \left( r_dT \left( s \left( \frac{1}{e} - 1 \right) + 1 \right) \right) \ge \left( 1 - \exp \left( T \left( s \left( \frac{1}{e} - 1 \right) + 1 \right) \right) \right)^K.
\]
\end{lemma}

\begin{proof}

\[
\begin{aligned}
1 - \exp \left( r_dT \left( s \left( \frac{1}{e} - 1 \right) + 1 \right) \right) &\ge \left( 1 - \exp \left( T \left( s \left( \frac{1}{e} - 1 \right) + 1 \right) \right) \right)^K\\
\ln\left(1 - \exp \left( r_dT \left( s \left( \frac{1}{e} - 1 \right) + 1 \right) \right) \right) &\ge K\ln \left( 1 - \exp \left( T \left( s \left( \frac{1}{e} - 1 \right) + 1 \right) \right) \right)\\
K &\ge \frac{\ln\left(1 - \exp \left( r_dT \left( s \left( \frac{1}{e} - 1 \right) + 1 \right) \right) \right)}{\ln \left( 1 - \exp \left( T \left( s \left( \frac{1}{e} - 1 \right) + 1 \right)  \right)\right)}
\end{aligned}
\]

This completes the proof.

\end{proof}

\section{Multi-bit Error Bound Analysis}\label{sec:multibit}

The dictionary-based methods \citep{kirchenbauer2023watermark, yoo2023advancing, wangtowards}, such as the multi-bit approach, partition the vocabulary into multiple blocks and identify watermarked text by determining the dominant partition using the maximum of several binomial variables. While prior work approximates this variable with a binomial distribution, our theoretical and numerical analysis in \Cref{sec:multibit} reveals that it instead follows a Gumbel distribution. As a result, the variable’s mean shifts with increasing capacity, exhibiting behavior akin to distribution-based methods. Our findings highlight that larger key capacities amplify the false detection problem. Numerical experiments further validate that as the message length grows, the detection statistic’s distribution shifts, making unwatermarked text increasingly indistinguishable from watermarked text. These insights emphasize the limitations of existing dictionary-based watermarking methods and the necessity of DW.

\citet{yoo2023advancing,wangtowards} extended \citet{kirchenbauer2023watermark}'s method to support multi-bit encoding. Their approach detects if a text is watermarked by use of a binomial statistic \citep{yoo2023advancing}. However, since the statistic is based on the maximal value of multiple binomial variables, it should no longer be considered a measure of a binomial distribution, but instead an approximate Gumbel distribution \citep{kotz2000extreme,haan2006extreme}. 

As the parameter for the Gumbel distribution is challenging to compute, we directly derive a novel bound for the composed extreme variable. Our analysis reveals that this method continues to suffer from the false detection problem.

We follow the notation in \citet{yoo2023advancing}, and use $[r]$ to denote the sequence of length $r$, \( [r] = [1, 2, \cdots, r] \). Given a generated sequence \([x_1, \cdots, x_T]\), \citet{yoo2023advancing} first uses a hash key to compute the position \( p_t \) of the message \( m \) for the \( t \)-th token, denoted as \( \rho_t = m[p_t] \), where \( p_t \in [b] \) and \( \rho_t \in [r] \). Here, \( b \) is the message length, and \( r \) indicates the number of bits each position encodes. Finally, the vocabulary $V$ is divided into \( r \) blocks \( [V_1, \cdots, V_r] \), and \( \delta \) is added to the logits of all tokens in the \( \rho_t \)-th partition \( V_{\rho_t} \).

When determining if a text is watermarked, the method calculates the maximal count in each vocabulary block for each position \( p_t \) normalized by the total count allocated to that position. For the \( p_t \)-th position, the random variable $C_{p_t}\in[0,1]$ can be denoted as:
\[
C_{p_t} = \max_{\rho \in [r]} \left\{ \frac{\sum_{t=1}^T \mathds{1}(x_t \in V_\rho) \cdot \mathds{1}(p_t = \rho)}{\sum_{t=1}^T \mathds{1}(p_t = \rho)} \right\}.
\]
Following \citet{yoo2023advancing}, we approximate the distribution for each block \( \rho \) using a binomial distribution. The total count for each block is approximated as \( \frac{T}{b} \). Therefore, we have:
\[
C_{p_t} = \max\left(\frac{X_1}{T/b}, \cdots, \frac{X_r}{T/b}\right), \quad \text{where} \quad X_\rho \sim \text{Binomial}\left(\frac{T}{b}, \frac{1}{r}\right).
\]
\citet{yoo2023advancing} claims that if the text is unwatermarked, \( C_{p_t} \approx \frac{1}{r} \). Based on this, the detection method tests if \( C_{p_t} \) exceeds a predefined threshold, classifying the text as watermarked if this is the case. However, \citet{yoo2023advancing}'s approach approximates the distribution of \( C_{p_t} \) with a binomial distribution. Since \( C_{p_t} \) is the maximum of i.i.d. distributions, it is, in fact, a Gumbel distribution. As a result, even when the text is not watermarked, \( C_{p_t} \) is still likely to exceed \( \frac{1}{r} \), leading to excess false positives.

To further demonstrate this issue, we provide a theoretical analysis of how the random variable \( C_{p_t} \) grows as the key capacity increases. Given the difficulty of computing the parameters of the Gumbel distribution, we further analyze its tail bounds to examine how the parameter \( b \) affects \( C_{p_t} \). We first present the following theorem:

\begin{theorem}\label{thm:mbit}
Let \( C_{p_t} = \max\left(\frac{X_1}{T/b}, \frac{X_2}{T/b}, \cdots, \frac{X_r}{T/b}\right) \), where \( X_\rho \sim \text{Binomial}\left(\frac{T}{b}, \frac{1}{r}\right) \) for all \( \rho \in [r] \). Then, the probability that \( C_{p_t} \) exceeds a threshold \( y \) is bounded by:
\[
\Pr(C_{p_t} \geq y) \leq r \cdot \exp\left(-\frac{2 T \left(y - \frac{1}{r}\right)^2}{b}\right).
\]
\end{theorem}

\begin{proof}
For each block \( \rho \in [r] \), the normalized count is \( \frac{X_\rho}{T/b} \), where \( X_\rho \sim \text{Binomial}\left(\frac{T}{b}, \frac{1}{r}\right) \). The expectation of \( \frac{X_\rho}{T/b} \) is:
\[
\mathbb{E}\left[\frac{X_\rho}{T/b}\right] = \frac{\mathbb{E}[X_\rho]}{T/b} = \frac{1}{r}.
\]

We aim to bound the probability \( \Pr\left(\frac{X_\rho}{T/b} \geq y\right) \). This is equivalent to:
\[
\Pr\left(\frac{X_\rho}{T/b} \geq y\right) = \Pr\left(X_\rho \geq y \cdot \frac{T}{b}\right).
\]

Using Hoeffding's inequality for \( X_\rho \), we have:
\[
\Pr\left(X_\rho \geq y \cdot \frac{T}{b}\right) \leq \exp\left(-\frac{2 \left(y \cdot \frac{T}{b} - \mu\right)^2}{T/b}\right),
\]
where \( \mu = \mathbb{E}[X_\rho] = \frac{T}{b} \cdot \frac{1}{r} \).

Substitute \( \mu \) into the inequality:
\[
\Pr\left(X_\rho \geq y \cdot \frac{T}{b}\right) \leq \exp\left(-\frac{2 \left(y \cdot \frac{T}{b} - \frac{T}{b} \cdot \frac{1}{r}\right)^2}{T/b}\right).
\]

Simplify the argument of the exponential:
\[
\Pr\left(\frac{X_\rho}{T/b} \geq y\right) \leq \exp\left(-\frac{2 T \left(y - \frac{1}{r}\right)^2}{b}\right).
\]

Now, for the maximum \( C_{p_t} = \max\left(\frac{X_1}{T/b}, \frac{X_2}{T/b}, \cdots, \frac{X_r}{T/b}\right) \), we use the union bound:
\[
\Pr(C_{p_t} \geq y) \leq \sum_{\rho=1}^r \Pr\left(\frac{X_\rho}{T/b} \geq y\right).
\]

Since the bound for each \( \rho \) is identical, we multiply the single block bound by \( r \):
\[
\Pr(C_{p_t} \geq y) \leq r \cdot \exp\left(-\frac{2 T \left(y - \frac{1}{r}\right)^2}{b}\right).
\]

This completes the proof.
\end{proof}

It can be observed from \Cref{thm:mbit} that as the message length \( b \) increases, the probability that \( C_{p_t} \) exceeds a certain threshold, \( \Pr(C_{p_t} \geq y) \), also increases. This implies that as the key capacity grows, the method becomes more prone to false detection problems. 
One might argue that increasing \( r \) can also increase the key capacity. However, it should be noted that as \( r \) increases significantly, the vocabulary will be divided into \( r \) blocks, causing the "green list" to become smaller and smaller. This reduction in the green list size makes it increasingly difficult to contain feasible next tokens, further complicating the watermarking process.

We further conducted a numerical experiment to demonstrate how the distribution shifts as \( b \) increases. The results are presented in \Cref{fig:multibitb}. We fix \( r=10 \), indicating that each position contains 10 bits of information, and vary the message length \( b \in [2, 20] \). Additionally, we plot the desired binomial distribution for \( r=10 \) using the red line, as expected in the original paper. The results demonstrate that (1) as the message length \( b \) increases, the expectation of the random variable \( C_{p_t} \) also rises. For the origional Multibit method, the expected value is 0.1, but it continues to grow as \( b \) increases, further validating the correctness of our theoretical analysis. This shift causes the unwatermarked text to resemble watermarked text, making it more challenging to distinguish them using a threshold. (2) Compared with the original binomial distribution, applying the max operation shifts the distribution to the right, resulting in a narrower distribution with reduced variance.

\begin{figure}[H]
    \centering
        \includegraphics[width=0.5\linewidth]{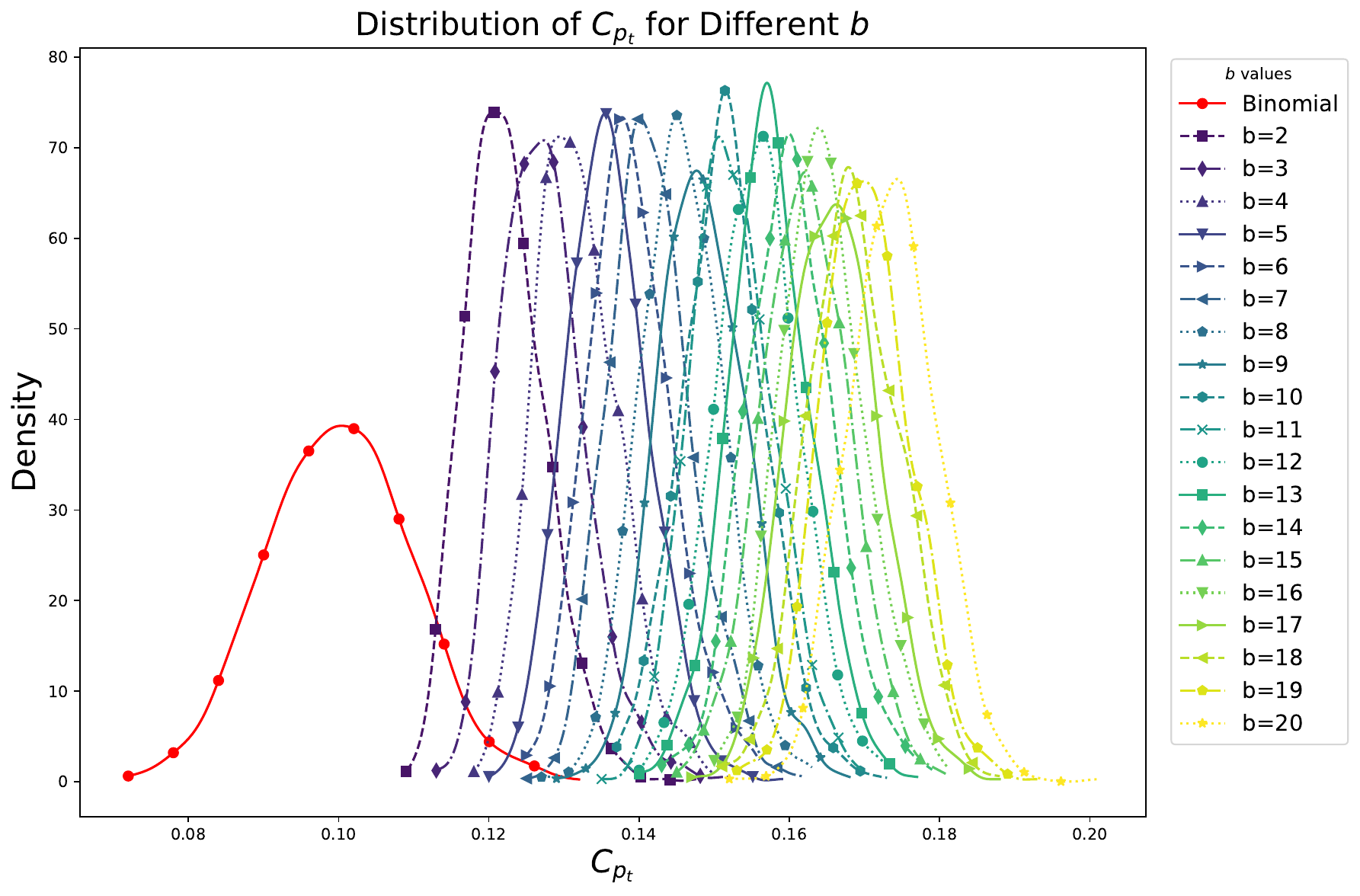}
        \vspace{0em}
        \caption{Distribution of $C_{p_t}$ with respect to different walues of $b$.}
        \label{fig:multibitb}
\end{figure}

\section{Numerical Bound Analysis} \label{sec:numerical}
Since the number of tokens used in the DW method is smaller than that used in FKE, we numerically calculated the minimal \(K\) value for which DW's bound outperforms FKE's bound, given specific \(T\) and \(r\). As shown in \Cref{tab:minK}, if \(K\) is larger than 114.7, the DW method's bound is superior for all settings we test. This capacity is relatively small, indicating that, our DW method outperforms the FKE method for almost any capacity \(K\). This conclusion is also supported by our analysis in \Cref{lemma:Kbound}.

\begin{table}[H]
        \centering
        \begin{tabular}{cccc}
        \toprule
        $T$ & $r = 0.2$ & $r = 0.5$ & $r = 0.8$ \\
        \hline
        200 & 9.3 & 3.6 & 1.6 \\
        300 & 21.1 & 6.0 & 2.0 \\
        400 & 48.7 & 10.2 & 2.5 \\
        500 & 114.7 & 17.7 & 3.1 \\
        \bottomrule
        \end{tabular}
        \captionof{table}{Lower bounds for \(K\) that DW is better than FKE with \(\tau_d=\tau_k = 1.6\) and varying \(T\) and \(r\).}
        \label{tab:minK}
\end{table}

\section{Watermarked Text Ratio}\label{sec:watermarktextratio}
The thresholds $\tau_k$ and $\tau_d$ are critical in determining whether a text is watermarked, with their optimal values varying according to the  watermarked text ratio. The metrics presented in the main experiment are averaged across datasets with differing ratios of watermarked text. To provide a comprehensive analysis of performance across various watermarked text ratios, we examine the relationship between the watermarked text ratio and the metrics Accu-I, Accu-O, and FPR. As illustrated in \Cref{fig:watermarkratio}, the following observations are made: (1) As the watermarked text ratio increases from 0\% to 100\%, Accu-I initially decreases and then rises after the 50\% mark. This trend occurs because, when the watermark ratio is 0\%, tuning the threshold on the development set to a very high value results in classifying all samples as unwatermarked, thereby leading to optimal performance. A similar situation arises when the watermark ratio nears 100\%, tuning the threshold to classify all samples as watermarked yields the best performance on both the development set and the test set. (2) Accu-O decreases as the watermark ratio increases, due to models ``overfitting'' to predict all samples as unwatermarked when the ratio is 0\%. As the ratio increases, predicting the exact key ID becomes more challenging than merely predicting whether the text is watermarked. However, Accu-O slightly increases as the watermark ratio approaches 100\% for FKE and PKE. This improvement occurs because the models are tuned to avoid predicting any samples as unwatermarked. (3) With an increasing watermarked text ratio, FPR also increases. This is because, at low watermark ratios, the tuned thresholds are set very high, making it unlikely for any sample to be classified as watermarked, thus eliminating false positive inferences. (4) It is evident that our proposed DW and HDW models outperform the FKE model across nearly all watermarked text ratios, demonstrating the effectiveness of our approach in various scenarios.

\begin{figure}[H]
    \centering
    \begin{minipage}{0.3\linewidth}
        \centering        \includegraphics[width=1.0\linewidth]{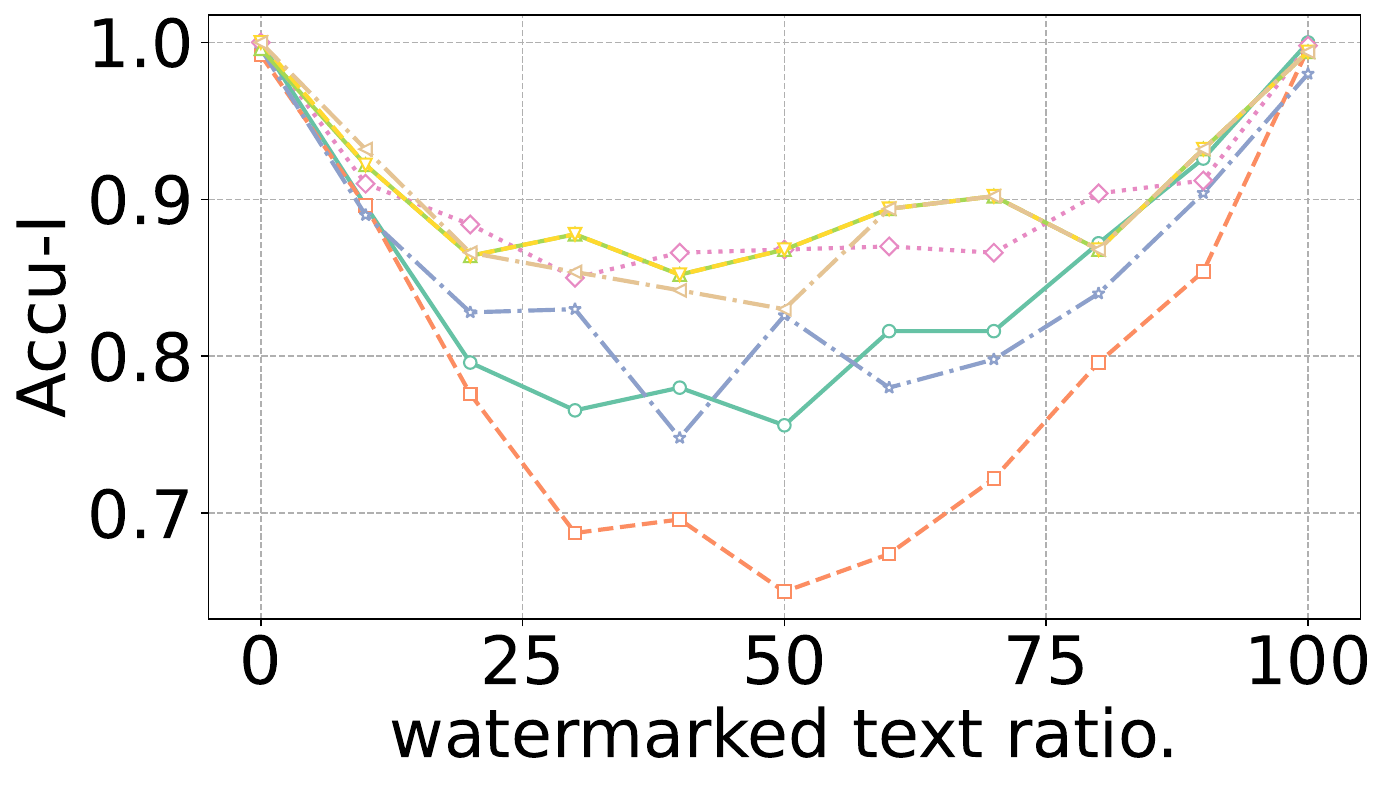}
        \label{fig:accud}
    \end{minipage}
    \hfill
    \begin{minipage}{0.3\linewidth}
        \centering
        \includegraphics[width=\linewidth]{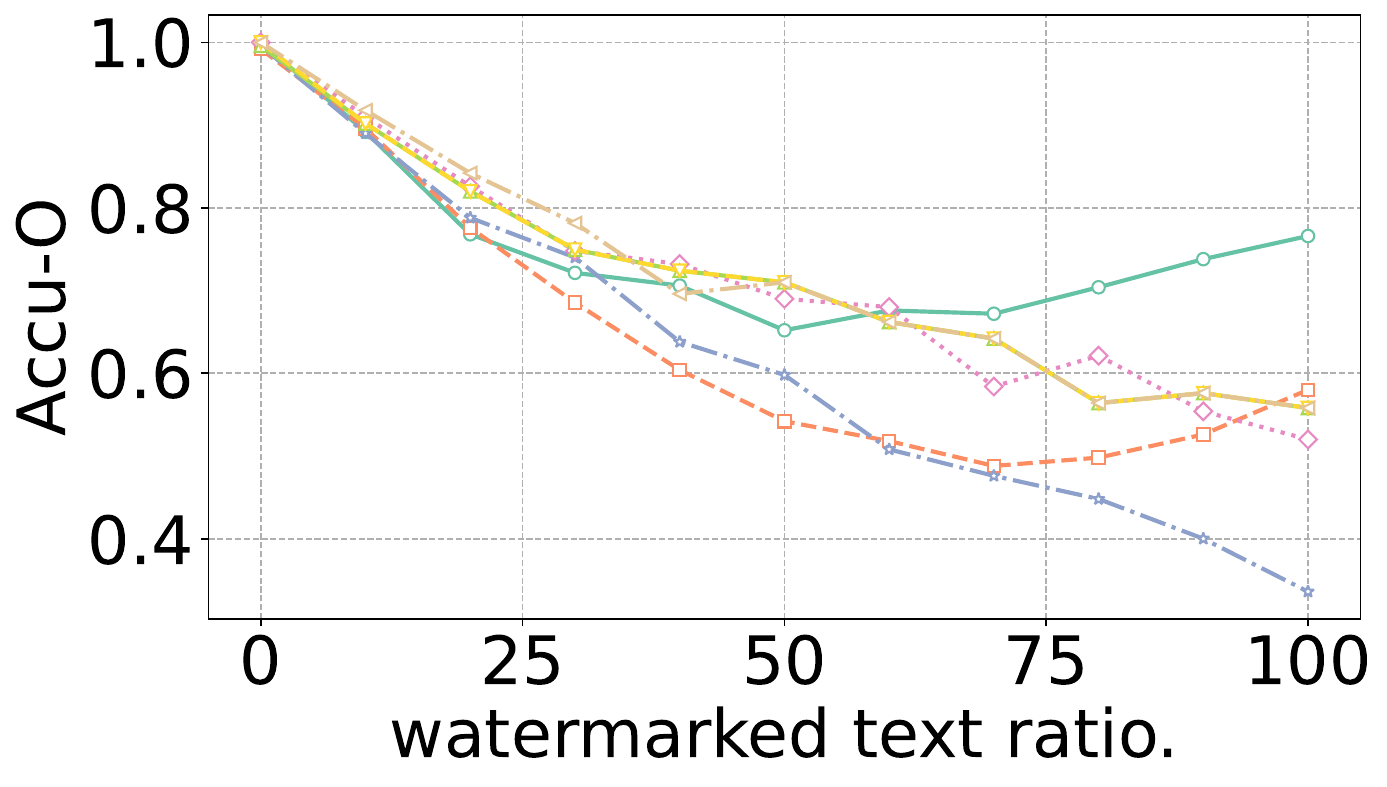}
        \label{fig:accuu}
    \end{minipage}
    \hfill
    \begin{minipage}{0.38\linewidth}
        \centering
        \includegraphics[width=\linewidth]{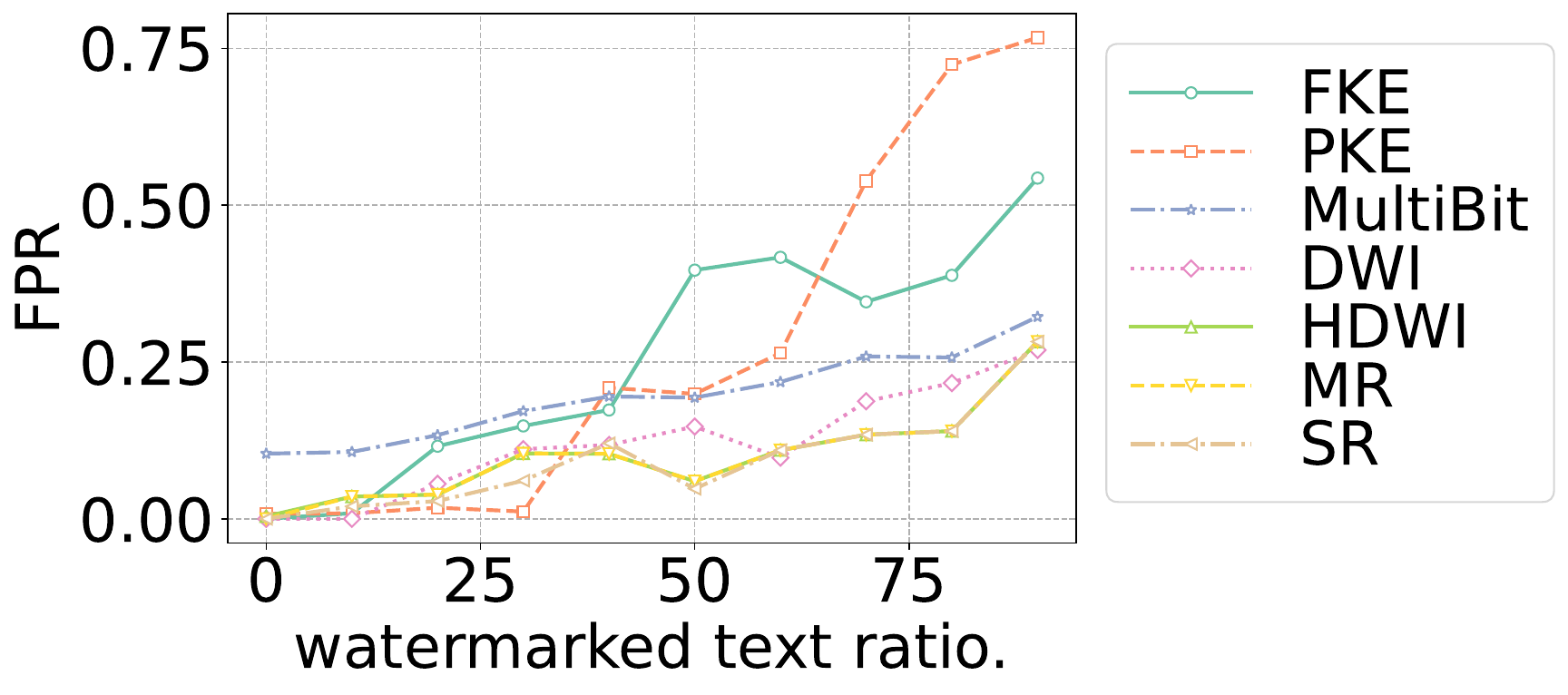}
        \label{fig:fpr}
    \end{minipage}
    \vspace{-1em}
    \caption{Watermarked text ratio results. The figures illustrate the relationship between the watermarked text ratio $r$ and the corresponding metrics. Each plot represents a specific metric, with metrics calculated by varying the thresholds $\tau_k$ and $\tau_d$ according to the watermark ratio $r$.}
    \label{fig:watermarkratio}
\end{figure}

\section{Watermarked Text Ratio for Dictionary-based method}\label{watermarktextratiodictionary}
Similar to \Cref{sec:watermarktextratio}, we conduct a watermarked text ratio analysis for the dictionary-based method using Multi-bit as the backbone. The results are presented in \Cref{fig:watermarkratiomaryland}. From the results, it can be observed that the trends for Accu-I, Accu-O, and FPR are similar to those observed in the experiments in \Cref{sec:watermarktextratio}, demonstrating that our method performs effectively in the dictionary-based approach. This observation further highlights the extensibility of our method across different backbones.

\begin{figure}[H]
    \centering
    \begin{minipage}{0.3\linewidth}
        \centering        \includegraphics[width=1.0\linewidth]{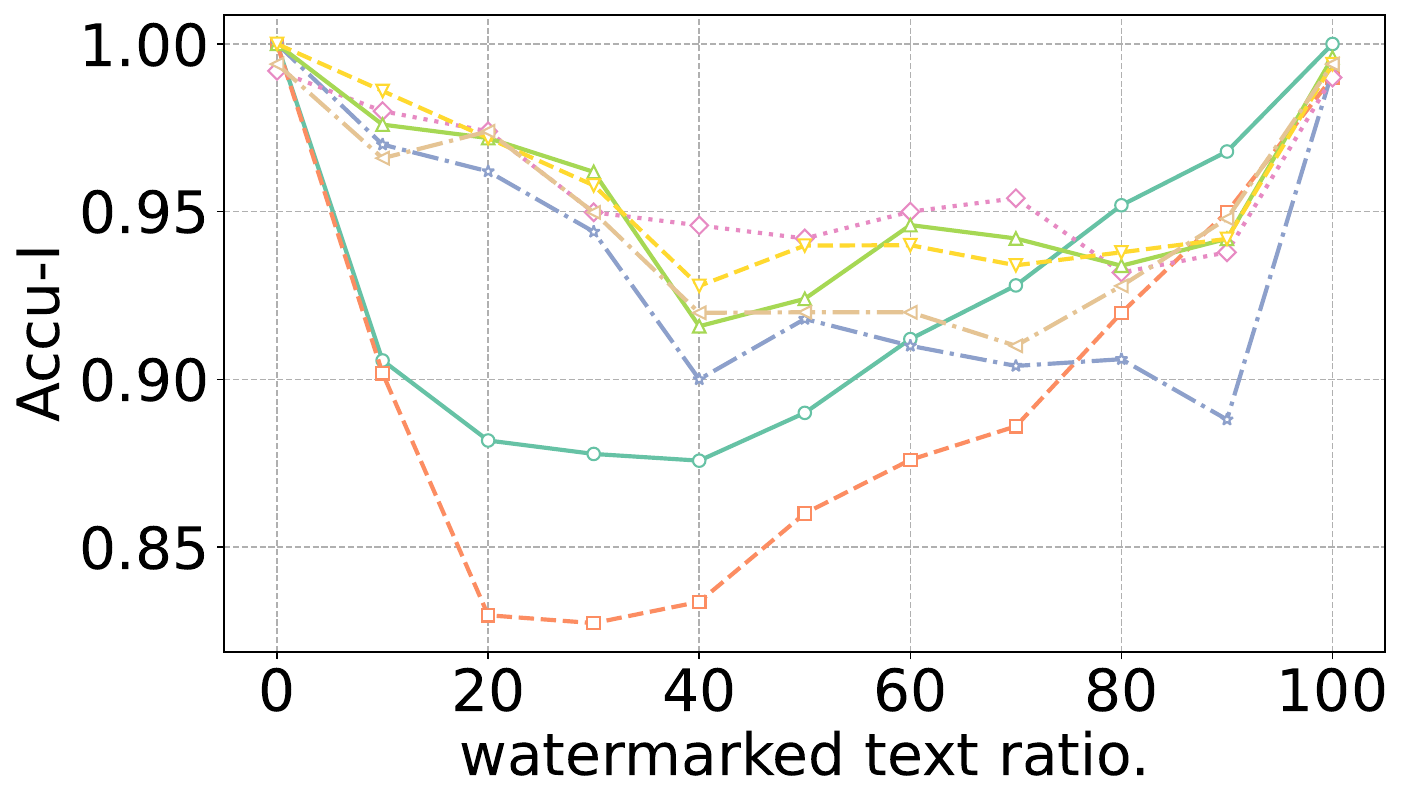}
        \label{fig:marylandaccud}
    \end{minipage}
    \hfill
    \begin{minipage}{0.3\linewidth}
        \centering
        \includegraphics[width=\linewidth]{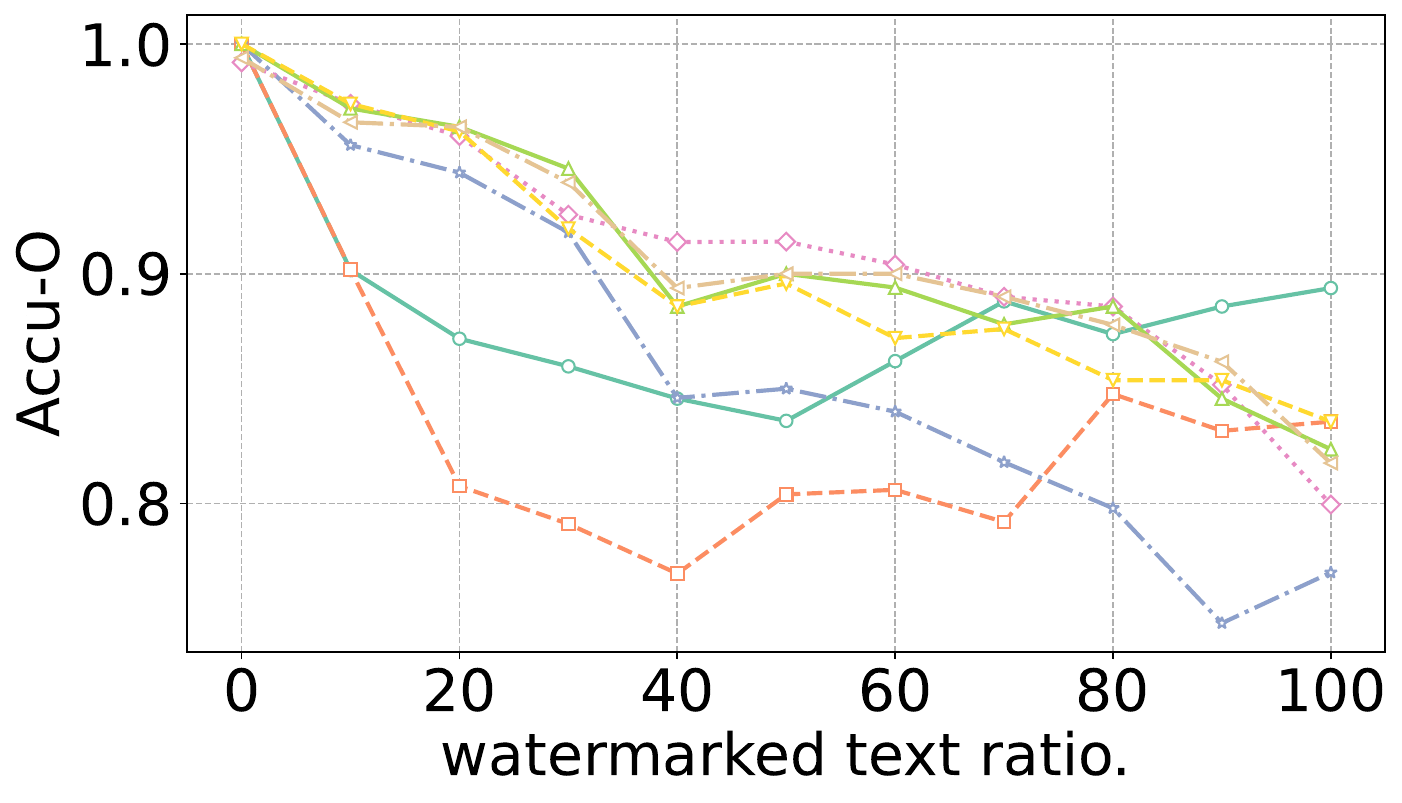}
        \label{fig:marylandaccuu}
    \end{minipage}
    \hfill
    \begin{minipage}{0.38\linewidth}
        \centering
        \includegraphics[width=\linewidth]{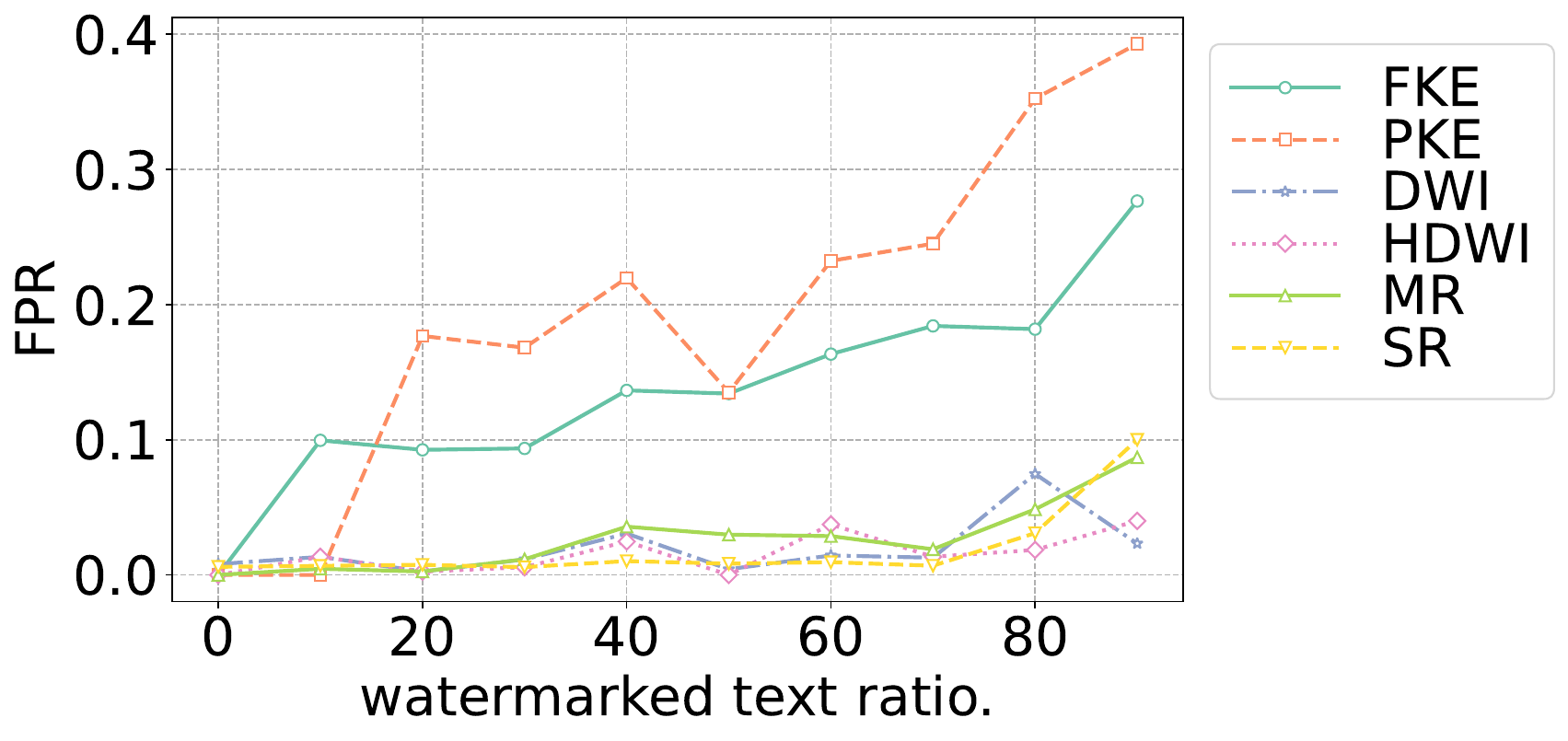}
        \label{fig:marylandfpr}
    \end{minipage}
    \vspace{-1em}
    \caption{Dictionary-based watermarked text ratio results. All models are based on Multi-bit. The figures illustrate the relationship between the watermarked text ratio $r$ and the corresponding metrics. Each plot represents a specific metric, with metrics calculated by varying the thresholds $\tau_k$ and $\tau_d$ according to the watermark ratio $r$.}
    \label{fig:watermarkratiomaryland}
\end{figure}

\section{Insertion and Deletion Attack}\label{sec:insdel}

Following \citet{fernandez2023three}, we also perform insertion and deletion attacks, randomly inserting or deleting tokens from the generated text to assess whether such modifications can effectively remove the watermark. We vary the insertion/deletion ratios in the range $[10\%,\cdots, 90\%]$. For instance, if the insertion ratio is 10\%, this indicates that we insert tokens amounting to 10\% of the total sequence length. Similarly, a deletion ratio of 10\% means removing 10\% of the tokens from the generated sequence. The experimental results are presented in \Cref{fig:insexp} and \Cref{fig:delexp} respectively.
The results indicate the following observations: (1) As the insertion/deletion ratio increases, all scores decrease. This is expected, as modifying more tokens introduces additional noise, making it increasingly difficult to classify the tokens. (2) Our proposed DW and HDW methods perform nearly identically to the original FKE method, demonstrating that our approach retains the same robustness capabilities as the original methods.

\begin{figure}[H]
    \centering
    \begin{minipage}{0.32\linewidth}
        \centering
        \includegraphics[width=\linewidth]{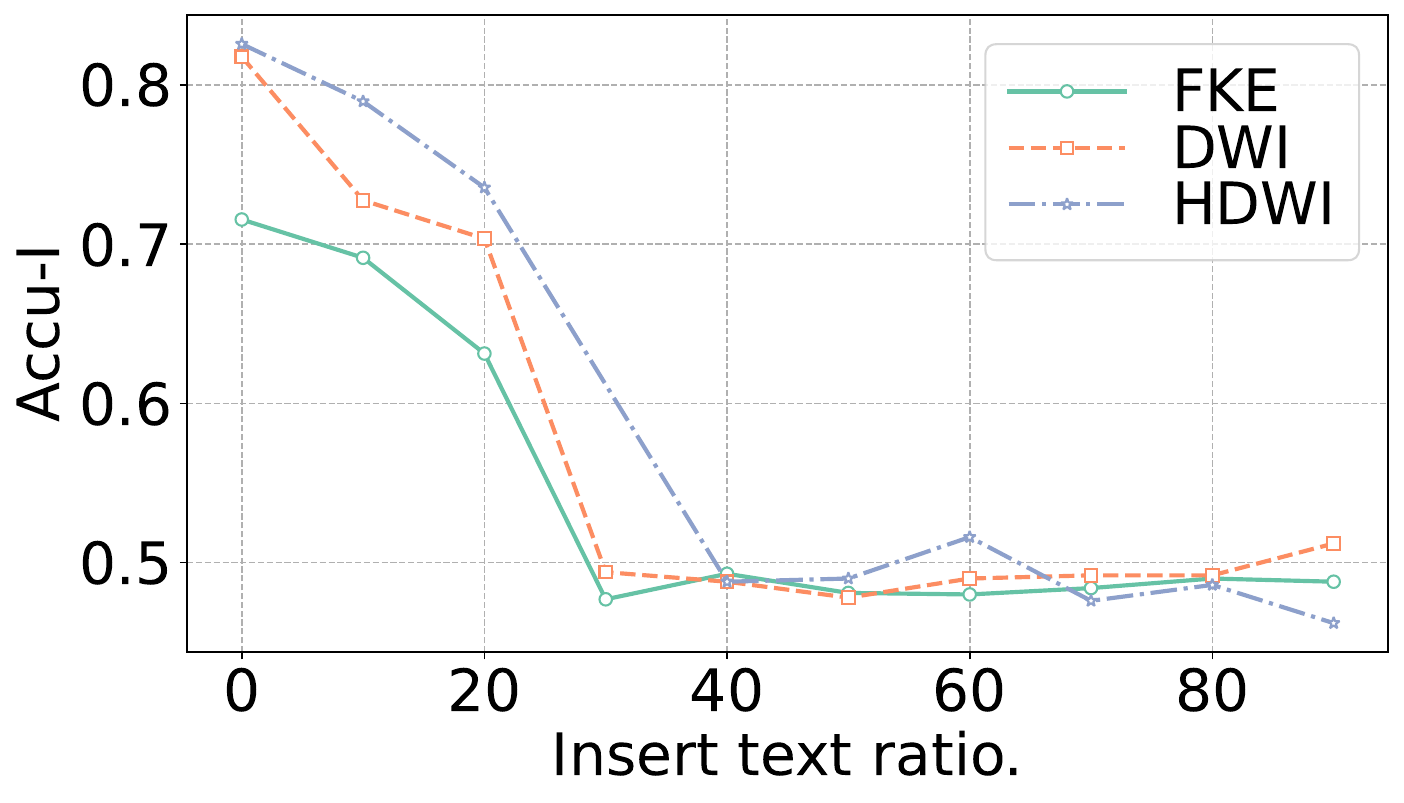}
    \end{minipage}
    \hfill
    \begin{minipage}{0.32\linewidth}
        \centering
        \includegraphics[width=\linewidth]{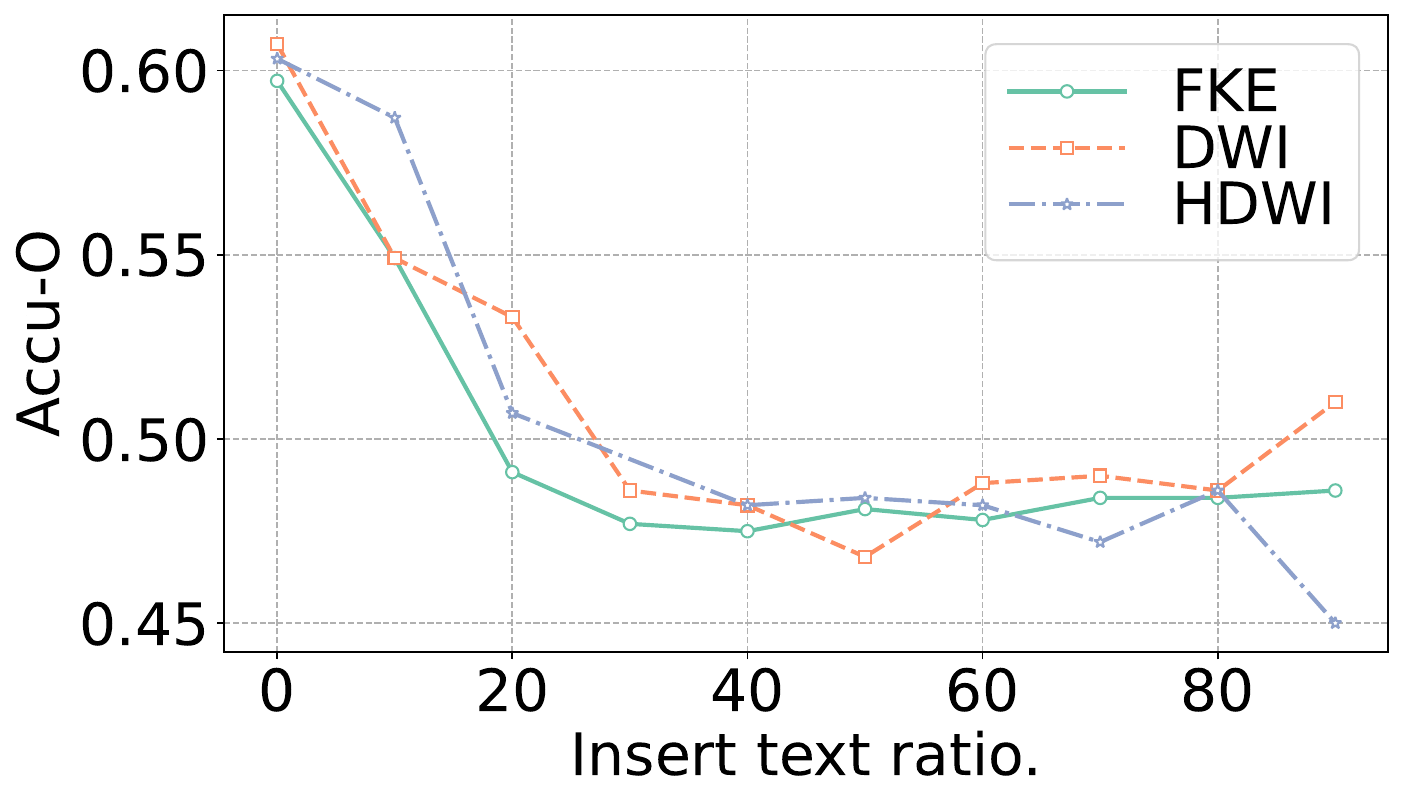}
    \end{minipage}
    \hfill
    \begin{minipage}{0.32\linewidth}
        \centering
        \includegraphics[width=\linewidth]{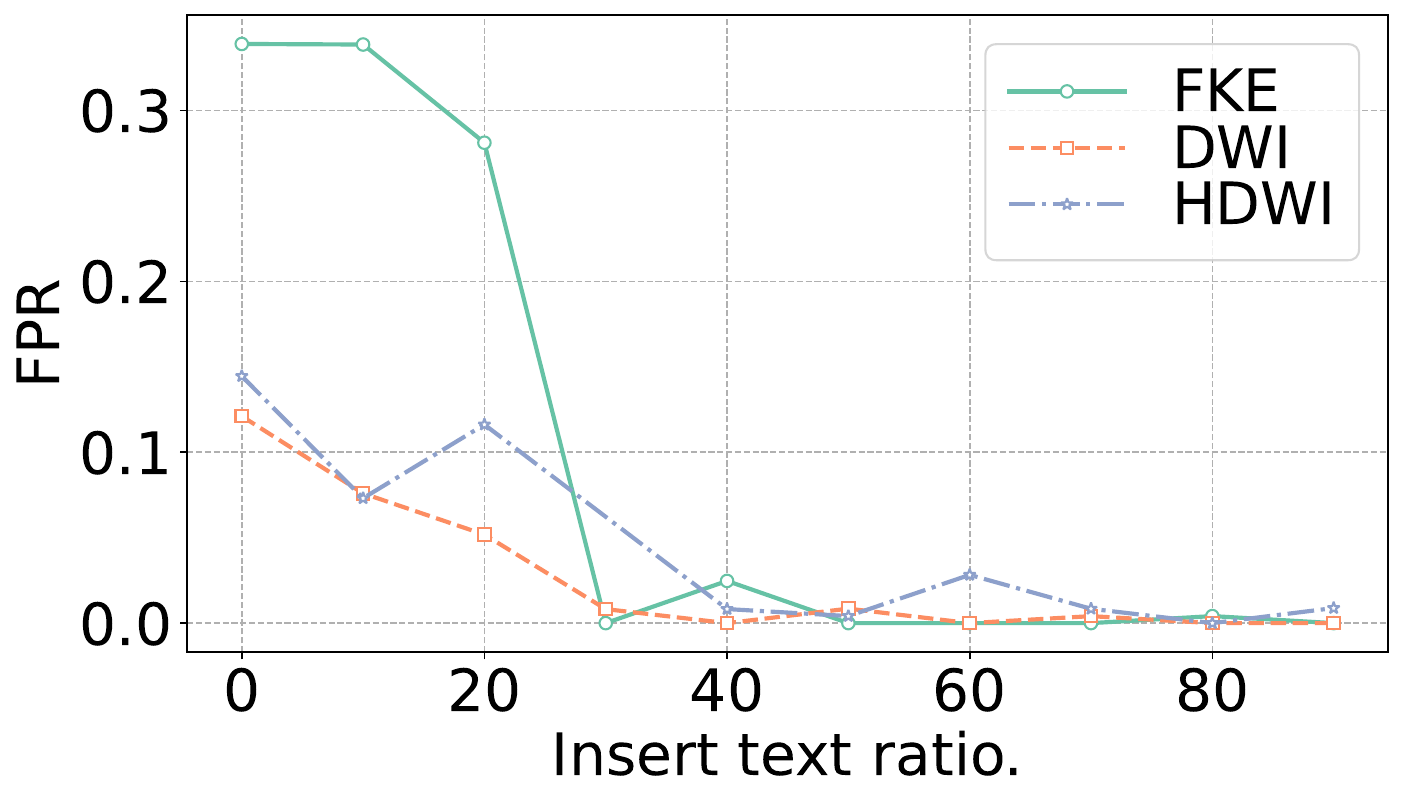}
    \end{minipage}
    \caption{Insertion attack results. The figure shows the impact of varying insertion ratios (10\% to 90\%) on the metrics Accu-I, Accu-O, and FPR for different watermarking methods (FKE, DW, HDW).}
    \label{fig:insexp}
\end{figure}

\begin{figure}[h]
    \centering
    \begin{minipage}{0.32\linewidth}
        \centering
        \includegraphics[width=\linewidth]{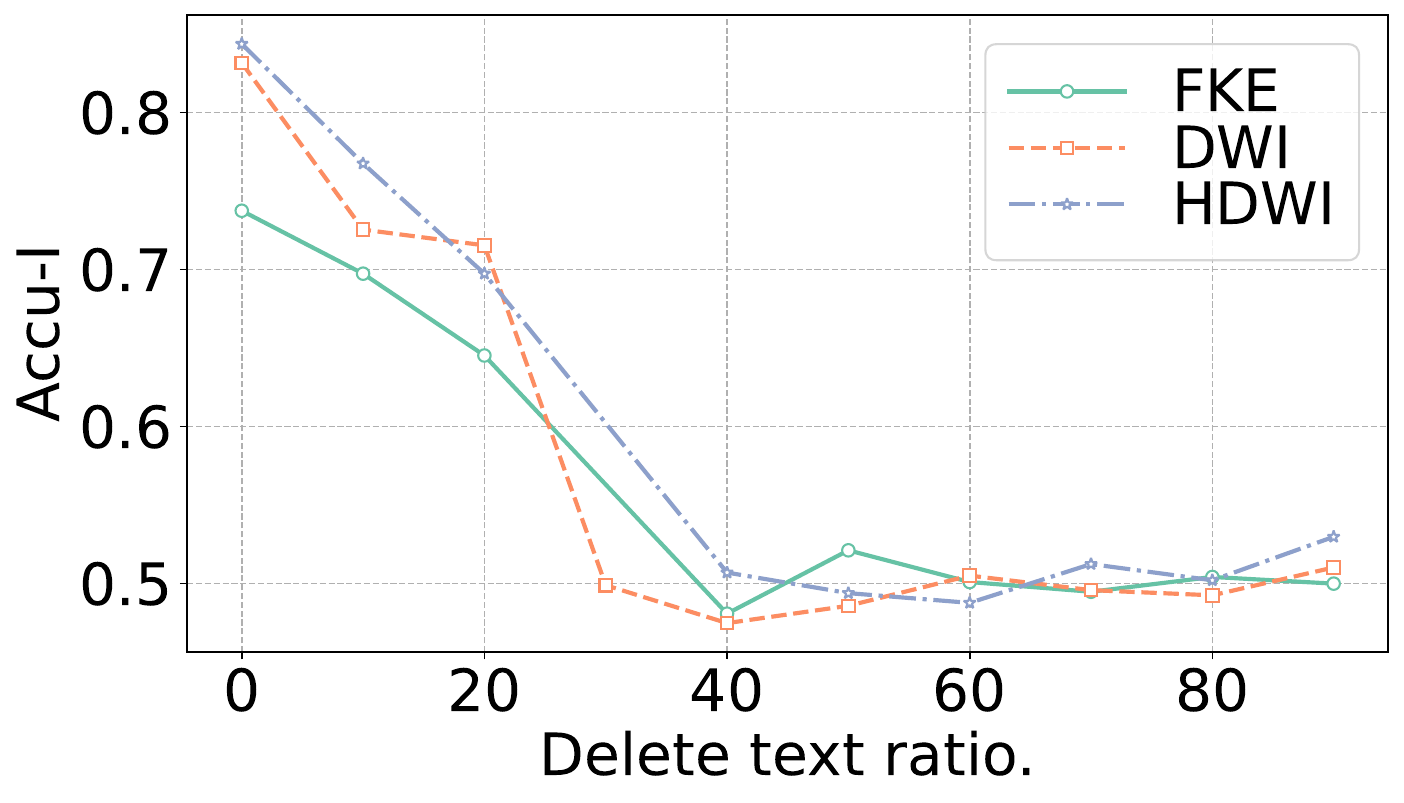}
    \end{minipage}
    \hfill
    \begin{minipage}{0.32\linewidth}
        \centering
        \includegraphics[width=\linewidth]{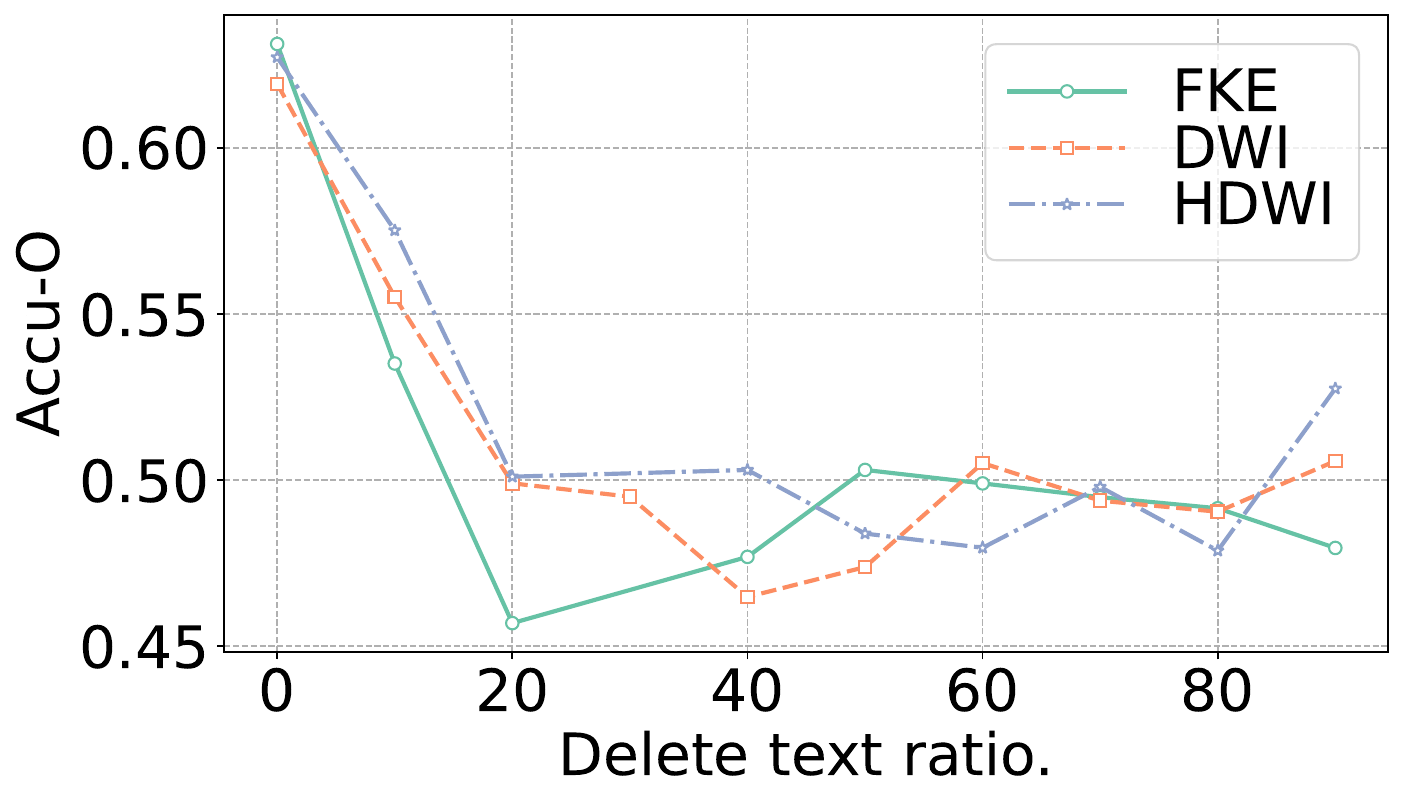}
    \end{minipage}
    \hfill
    \begin{minipage}{0.32\linewidth}
        \centering
        \includegraphics[width=\linewidth]{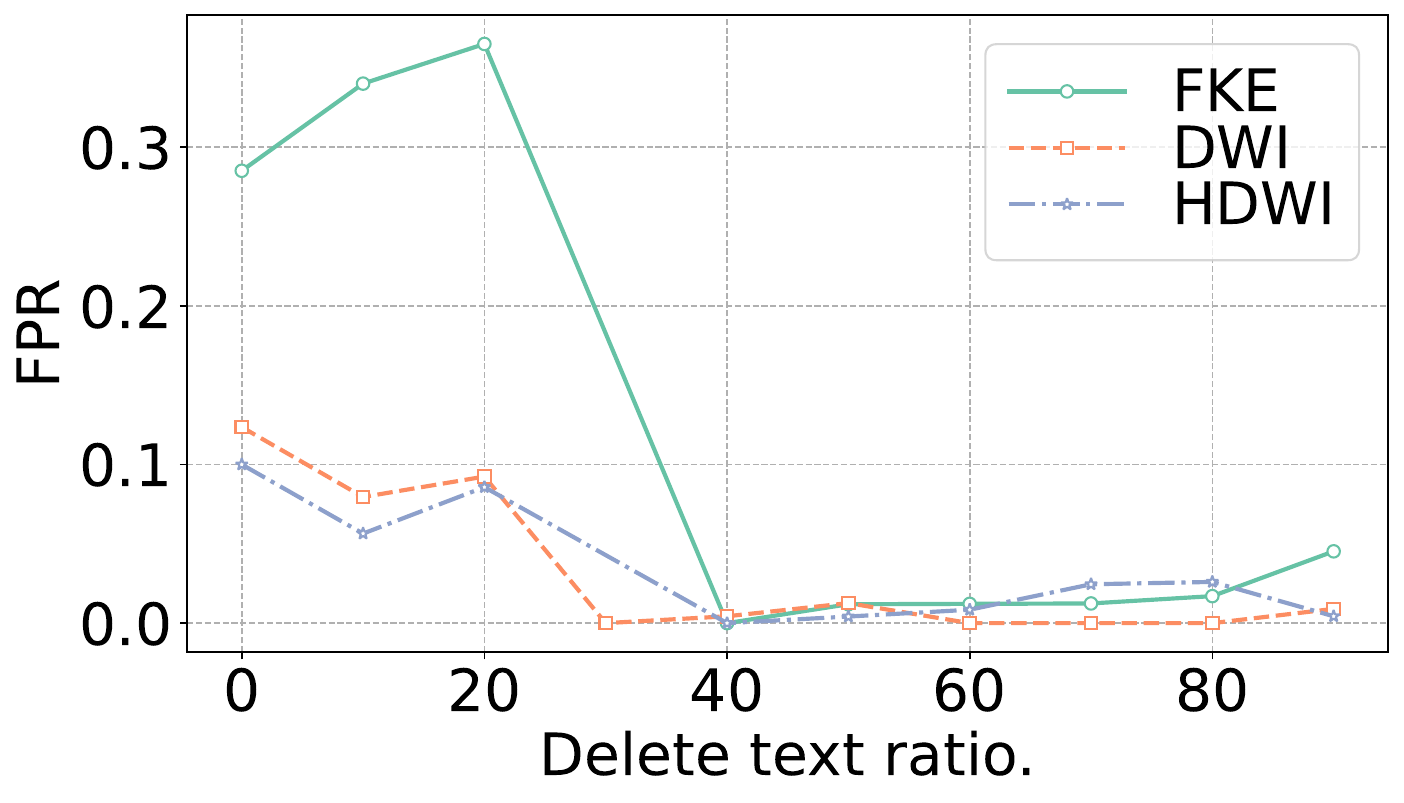}
    \end{minipage}
    \caption{Deletion attack results. The figure illustrates the effect of varying deletion ratios (10\% to 90\%) on the metrics Accu-I, Accu-O, and FPR for different watermarking methods (FKE, DW, HDW).}
    \label{fig:delexp}
\end{figure}

\section{Paraphrase Attack}\label{sec:paraphraseattack}
We conduct a paraphrase attack to evaluate the robustness of the proposed methods. We set a watermarked ratio $r=0.5$ to test whether the models can differentiate watermarked text. We use  Parrot\_Paraphraser\footnote{\url{https://github.com/PrithivirajDamodaran/Parrot_Paraphraser}}, a toolkit designed to rephrase sentences generated with watermarks, and we use the same detection tool to detect the watermark and key information. The results are shown in \Cref{fig:paraphraseattack}. We can observe that (1) our proposed DW and HDW models outperform the FKE method, (2) although accuracy decreases after the paraphrase attack, it remains above 0.5, indicating that the methods can still recognize watermarked text and associated keys, and (3) the FPR decreases after the attack because the models are more likely to classify text as unwatermarked. This outcome is expected because, after the paraphrase attack, some previously watermarked text can no longer be detected.

\begin{figure}[h]
    \centering
        \includegraphics[width=\linewidth]{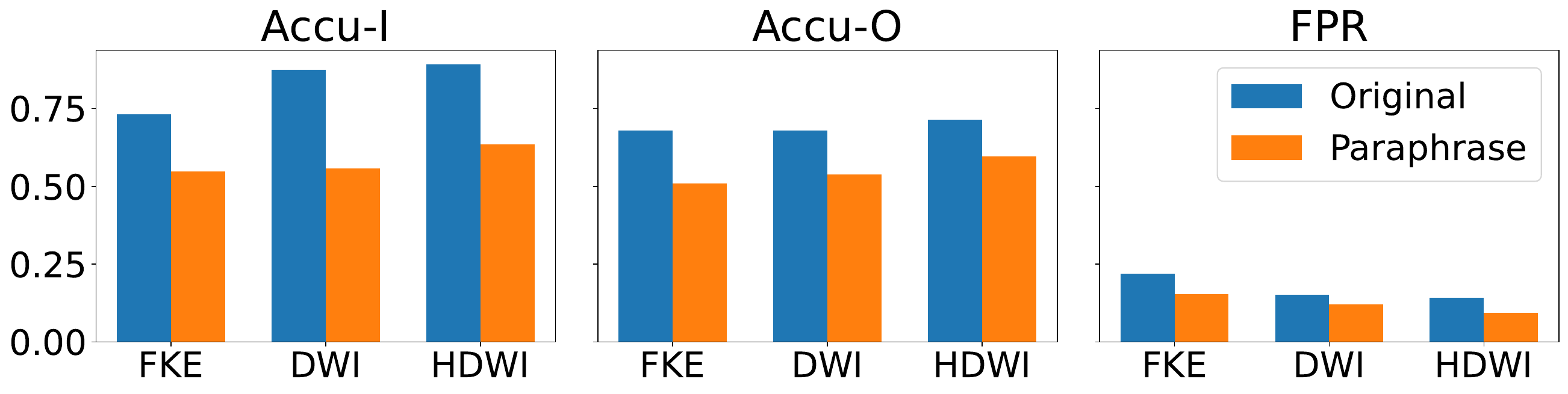}
        \vspace{0em}
        \caption{Paraphrase attack results. The figure compares the performance of watermarking methods (FKE, DW, HDW) on original and paraphrased text, showing metrics Accu-I, Accu-O, and FPR for a watermarked ratio $r = 0.5$ . }
        \label{fig:paraphraseattack}
\end{figure}

\section{Runtime Analysis}\label{sec:runtimeanalysis}
To evaluate the computational cost of the watermarking models, we conducted a runtime analysis experiment by testing the runtime for the same 100 samples across different models, including the “NoWatermark” generation. The results, shown in \Cref{fig:runtime}, reveal that (1) in the generation phase, the runtime for watermarking models is slightly higher than for non-watermarked generation, as additional time is required for hashing and key encoding, (2) the runtime for all models during the generation phase is independent of the key capacity, since the watermark encoding process only runs once and does not depend on the size of the key capacity, and (3) in the detection phase, the runtime for all watermarking models increases linearly with the key capacity, as the detection process involves multiple iterations over possible keys to identify the best matching key.

\begin{figure}[h]
    \centering
    \begin{minipage}{0.45\linewidth}
        \centering
        \includegraphics[width=\linewidth]{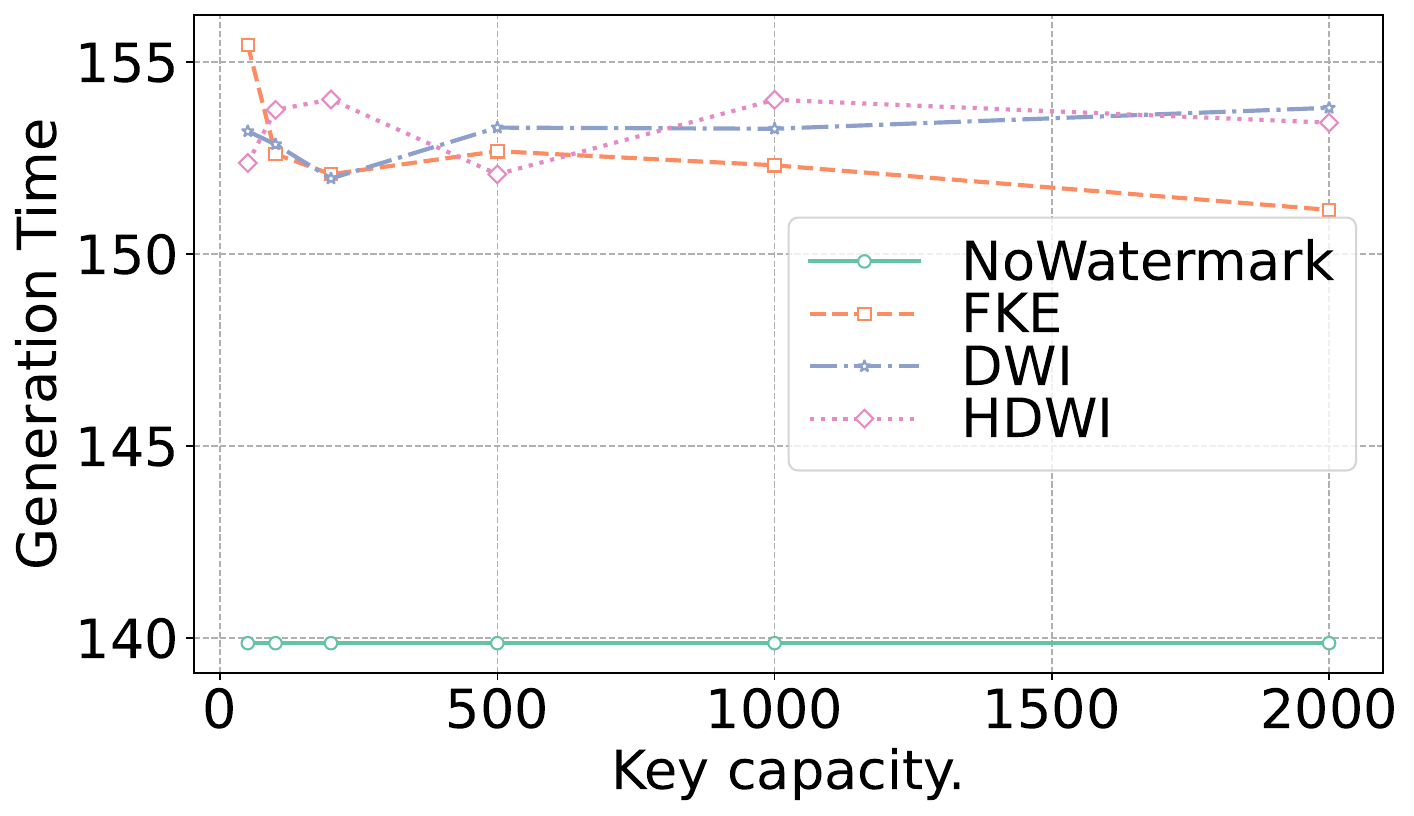}
    \end{minipage}
    \begin{minipage}{0.45\linewidth}
        \centering
        \includegraphics[width=\linewidth]{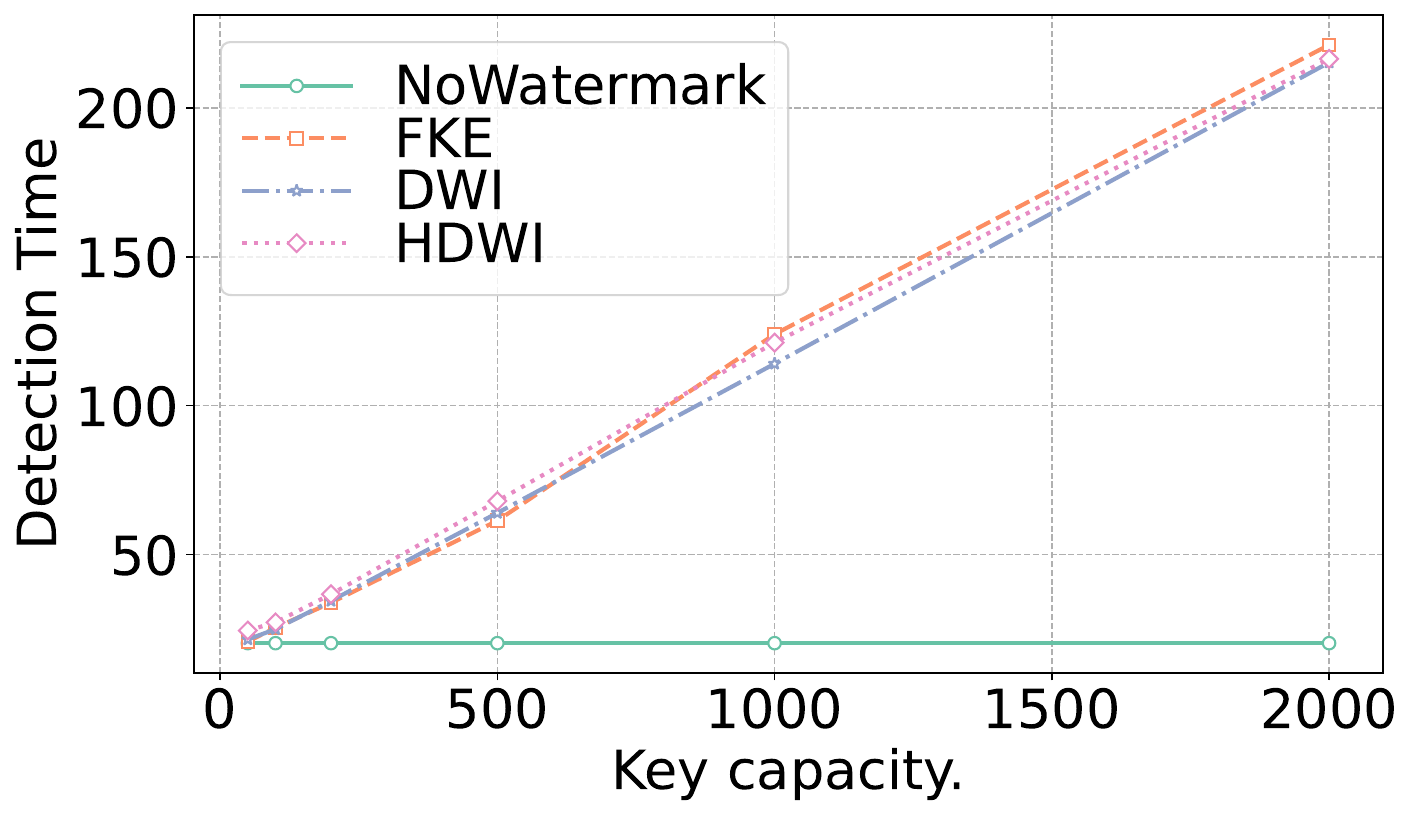}
    \end{minipage}
    \label{fig:runtime}
    \caption{Runtime analysis. The left plot shows the generation time, while the right plot shows the detection time for various watermarking models (NoWatermark, FKE, DW, HDW) as a function of key capacity ( K  ranging from 0 to 2000).}
\end{figure}

\section{Sequence Length Analysis}\label{sec:seqlenanalysis}

To evaluate how performance is influenced by the length of generated sequences, we conducted a sequence length analysis experiment using the HDW model with a watermark ratio $r=0.5$. The experiment tested sequence lengths ranging from 20 to 1000, and the results are presented in \Cref{fig:lenexp}. The following observations can be made: (1) as the sequence length increases, the accuracy scores improve, as longer sequences allow for clearer embedding of the watermark into the generated text, (2) as the sequence length grows, the FPR metric also increases; however, this does not necessarily indicate worsening false detection problems. When the text length is short, the model rarely recognizes any sequence as watermarked, leading to accuracy scores close to 0.5 and FPR close to 0. As the sequence length increases, the predicted positive rate rises, resulting in more false positives, and (3) based on the experiment, the models begin to recognize watermarked text effectively when the token length exceeds 50, and they achieve good performance when the token length exceeds 200.

\begin{figure}[h]
    \centering
        \includegraphics[width=0.5\linewidth]{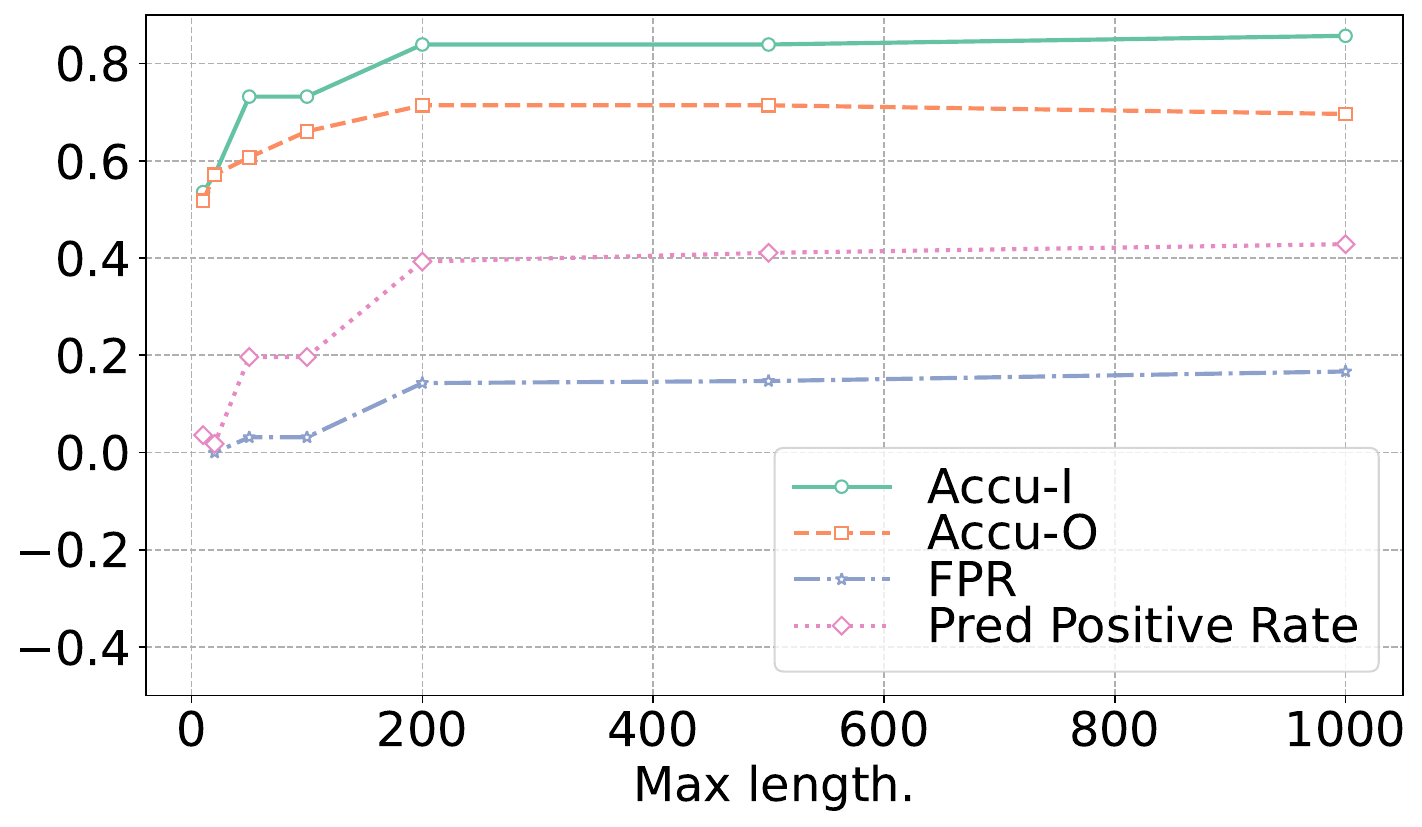}
        \vspace{0em}
        \caption{Sequence Length Analysis. The figure presents the impact of sequence length (ranging from 20 to 1000) on metrics (Accu-I, Accu-O, FPR, and Predicted Positive Rate) for the HDW model with a watermark ratio $r = 0.5$.}
        \label{fig:lenexp}
\end{figure}

\section{Experiments with More Datasets}\label{sec:moredata}

To demonstrate the applicability of our model across different scenarios, we conducted experiments on two domain-specific datasets: a biomedical question dataset, BioASQ \citep{krithara2023bioasq}, and a legal dataset, LegalQA\footnote{\url{https://huggingface.co/datasets/dzunggg/legal-qa-v1}}. We evaluated our models on these datasets, and the results are presented in \Cref{tab:bioasq} and \Cref{tab:legalqa}. The findings show that (1) the performance trends on these domain-specific datasets are generally consistent with those in the main experiment, with our proposed methods achieving superior results compared to other models, and (2) the similarity scores in both datasets are as high as 0.9, indicating that the watermarking method minimally alters the output text, even in highly specific domains.

\begin{table}[h]
    \centering
    \scriptsize
    \begin{minipage}{0.4\textwidth}
        \centering
        \scriptsize
        
        \begin{tabular}{lllll}
        \toprule
         & Accu-I$\uparrow$ & Accu-O$\uparrow$ & FPR$\downarrow$ & Sim$\uparrow$ \\
        \midrule
        FKE & 0.883 & 0.834 & 0.175 & 0.934 \\
        PKE & 0.797 & 0.679 & 0.20 & 0.931 \\
        MultiBit & 0.893 & 0.655 & 0.0825 & 0.922 \\
        DW & 0.963 & 0.744 & 0.0459 & 0.932 \\
        HDW & 0.923 & 0.73 & 0.0517 & 0.923 \\
        MR & 0.902 & 0.739 & 0.0769 & 0.923 \\
        SR & 0.963 & 0.773 & 0.0428 & 0.923 \\
            \bottomrule
        \end{tabular}
        
        \caption{BioASQ dataset results.}
        \label{tab:bioasq}
    \end{minipage}%
    \quad\quad\quad
    \begin{minipage}{0.4\textwidth}
        \centering
        \scriptsize
        
        \begin{tabular}{lllll}
        \toprule
         & Accu-I$\uparrow$ & Accu-O$\uparrow$ & FPR$\downarrow$ & Sim$\uparrow$ \\
        \midrule
        FKE & 0.955 & 0.95 & 0.0877 & 0.901 \\
        PKE & 0.873 & 0.80 & 0.0657 & 0.905 \\
        MultiBit & 0.941 & 0.732 & 0.175 & 0.89 \\
        DW & 0.955 & 0.831 & 0.118 & 0.905 \\
        HDW & 0.973 & 0.863 & 0.00952 & 0.892 \\
        MR & 0.946 & 0.846 & 0.0275 & 0.892 \\
        SR & 0.943 & 0.832 & 0.0906 & 0.892 \\
        \bottomrule
        \end{tabular}
        
        \caption{LegalQA dataset results.}
        \label{tab:legalqa}
    \end{minipage}
\end{table}

\begin{table}[t]
    
\end{table}

\begin{table}[t]
    
\end{table}

\section{Experiments with Indication Ratio  $r_d$}\label{sec:indicationratioparam}
The indication ratio parameter $r_d$ controls the ratio between tokens used to encode the indicator variable and those used to encode key information. We conduct an experiment to evaluate how different values of $r_d$ affect the results, as shown in \Cref{fig:rdexp}. The findings are summarized as follows: 

(1) As $r_d$ increases, both DW and HDW exhibit an improvement in the Accu-I score. This demonstrates that using more tokens to encode the indicator variable enhances the accuracy of detecting whether the text is watermarked, thereby validating the correctness of our proposed method and theoretical analysis.
(2) With an increase in $r_d$, the Accu-O score initially increases and then decreases. At smaller values of $r_d$, the performance improves as more tokens are available to detect whether the text is watermarked. However, when $r_d$ becomes too large, it impairs the detection of key information, leading to a decline in overall performance.
(3) DW performs worse in both Accu-I and Accu-O when $r_d$ is small. This occurs because DW does not reuse key information to detect whether the text is watermarked, giving HDW an advantage at smaller $r_d$ values. This further underscores the effectiveness of the HDW method.
(4) As $r_d$ increases, the FPR decreases. Allocating more tokens to encode the indicator variable helps alleviate the false positive problem, improving overall robustness.

\begin{figure}[H]
    \centering
    \begin{minipage}{0.32\linewidth}
        \centering
        \includegraphics[width=\linewidth]{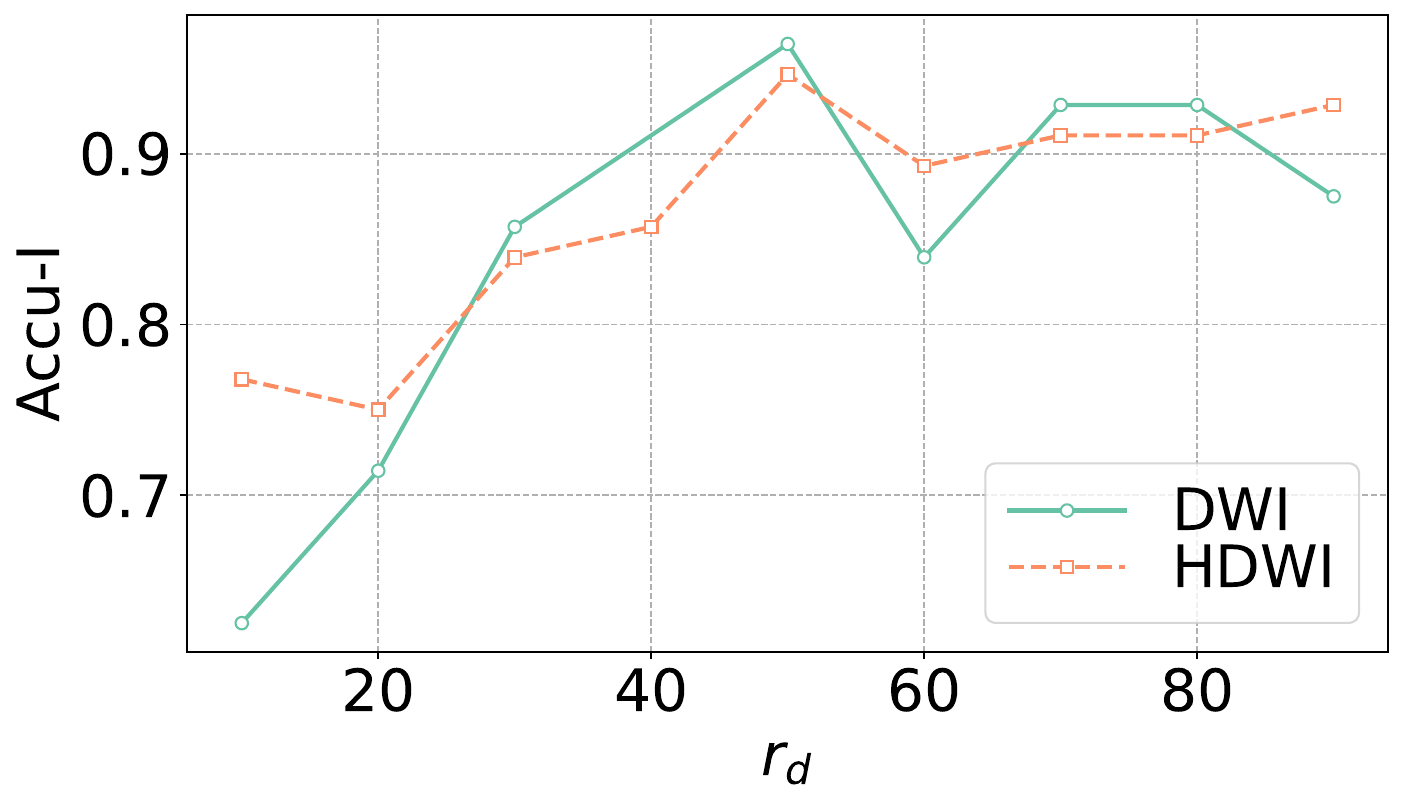}
    \end{minipage}
    \hfill
    \begin{minipage}{0.32\linewidth}
        \centering
        \includegraphics[width=\linewidth]{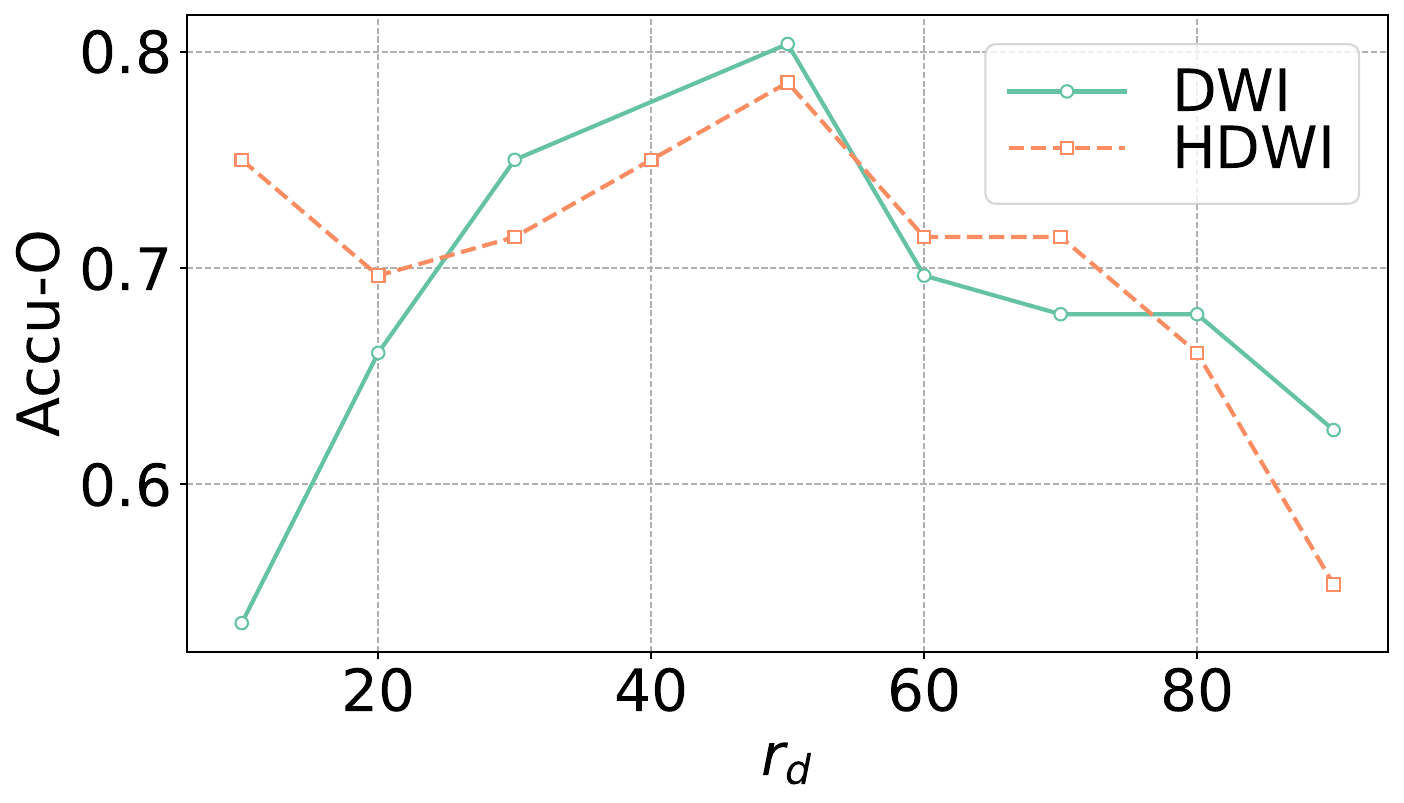}
    \end{minipage}
    \hfill
    \begin{minipage}{0.32\linewidth}
        \centering
        \includegraphics[width=\linewidth]{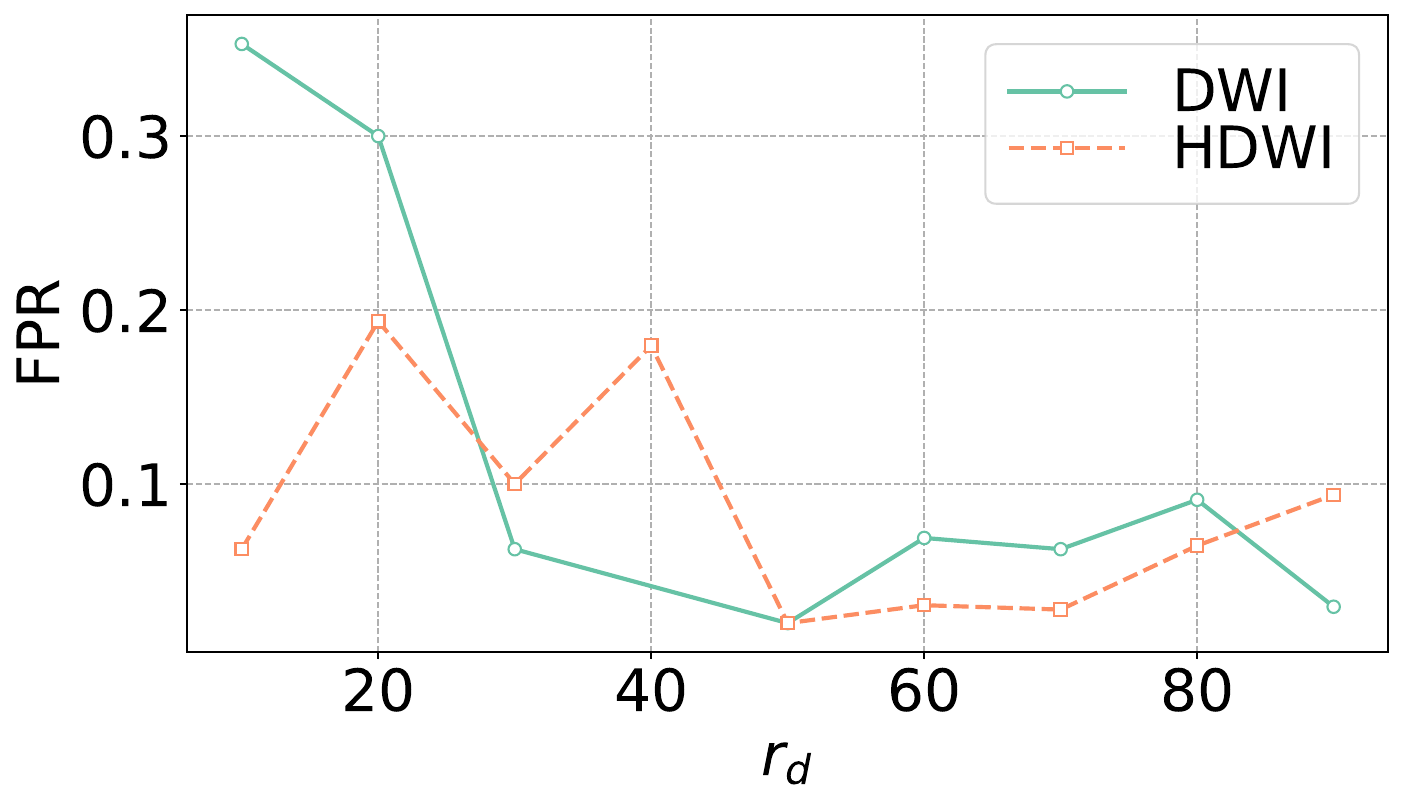}
    \end{minipage}
    \caption{Experiment results with varying indication ratio parameter $r_d$. Higher $r_d$ corresponds to more tokens allocated for encoding the indicator variable. }
    \label{fig:rdexp}
\end{figure}

\section{Experiments with Indication Ratio $r_d$ for Different Watermarked Text Ratio}\label{sec:indicationratio}
In the main experiment, we set $r_d=50\%$, thereby utilizing half of the tokens to encode whether the text is watermarked, while the remaining tokens encode the identification watermark. In this subsequent experiment, we evaluate the ratios $r_d$ in $[10\%, 30\%, 50\%, 70\%, 90\%]$ to assess the impact of employing more tokens to encode the watermark indicator. The results, depicted in \Cref{fig:varr}, indicate that (1) employing more tokens for encoding the indicator can substantially mitigate the FPR. (2) It can be observed that as the watermarked text ratio increases, the FPR also rises. This phenomenon occurs because, during threshold tuning on development set, when the ratio of watermarked text approaches zero, the model tends to select a threshold that directly classifies all samples as unwatermarked and thus make less false positive error. When the  watermarked text ratio is high, the threshold is adjusted to classify more samples as watermarked. This setting leads to an increase in false positive errors.

\begin{figure}[h]
    \centering
    \begin{minipage}[h]{0.45\textwidth}
        \includegraphics[width=0.99\linewidth]{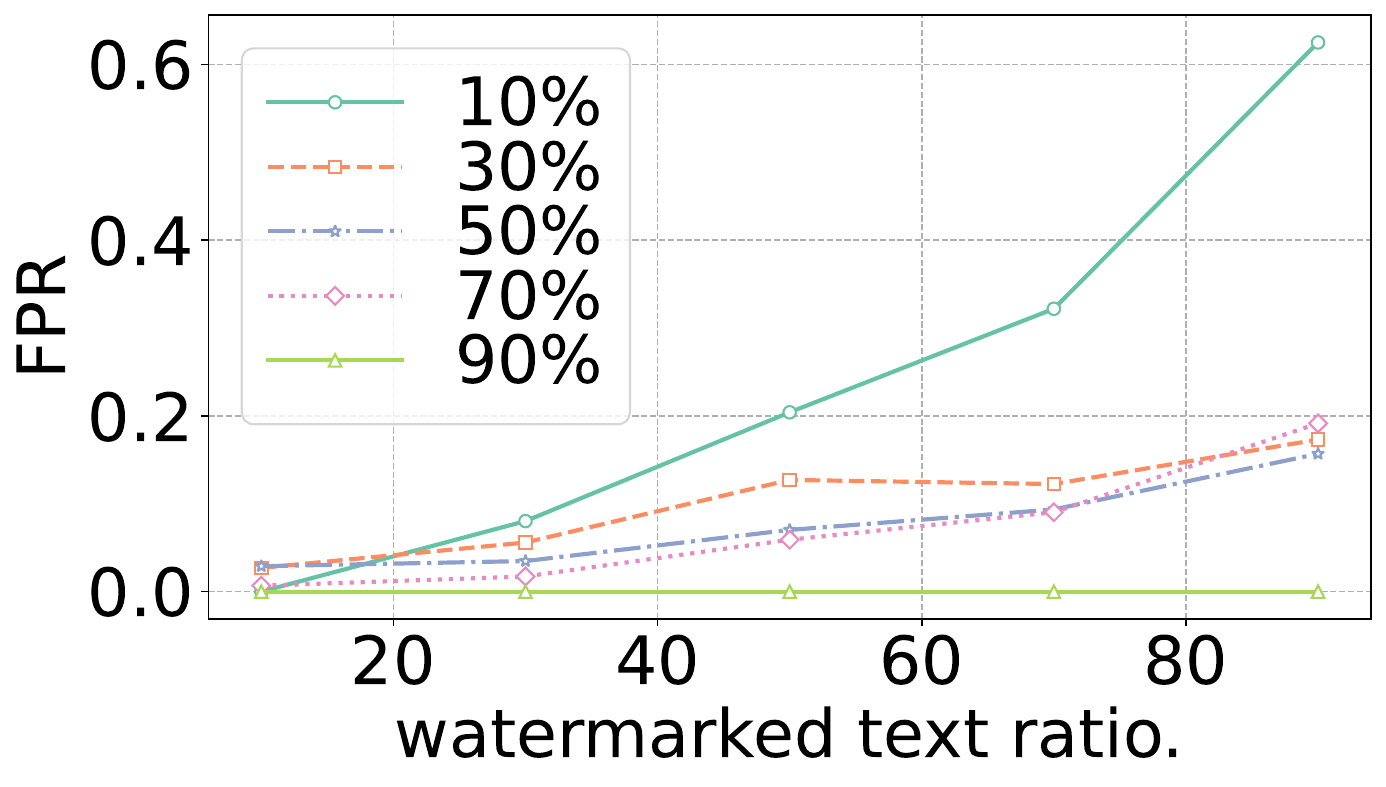}
        \captionof{figure}{Relationship between FPR and $r$. Each curve represents a specific $r$ value.}
        \label{fig:varr}
    \end{minipage}
    \quad
    \begin{minipage}[h]{0.45\textwidth}
        \includegraphics[width=0.99\linewidth]{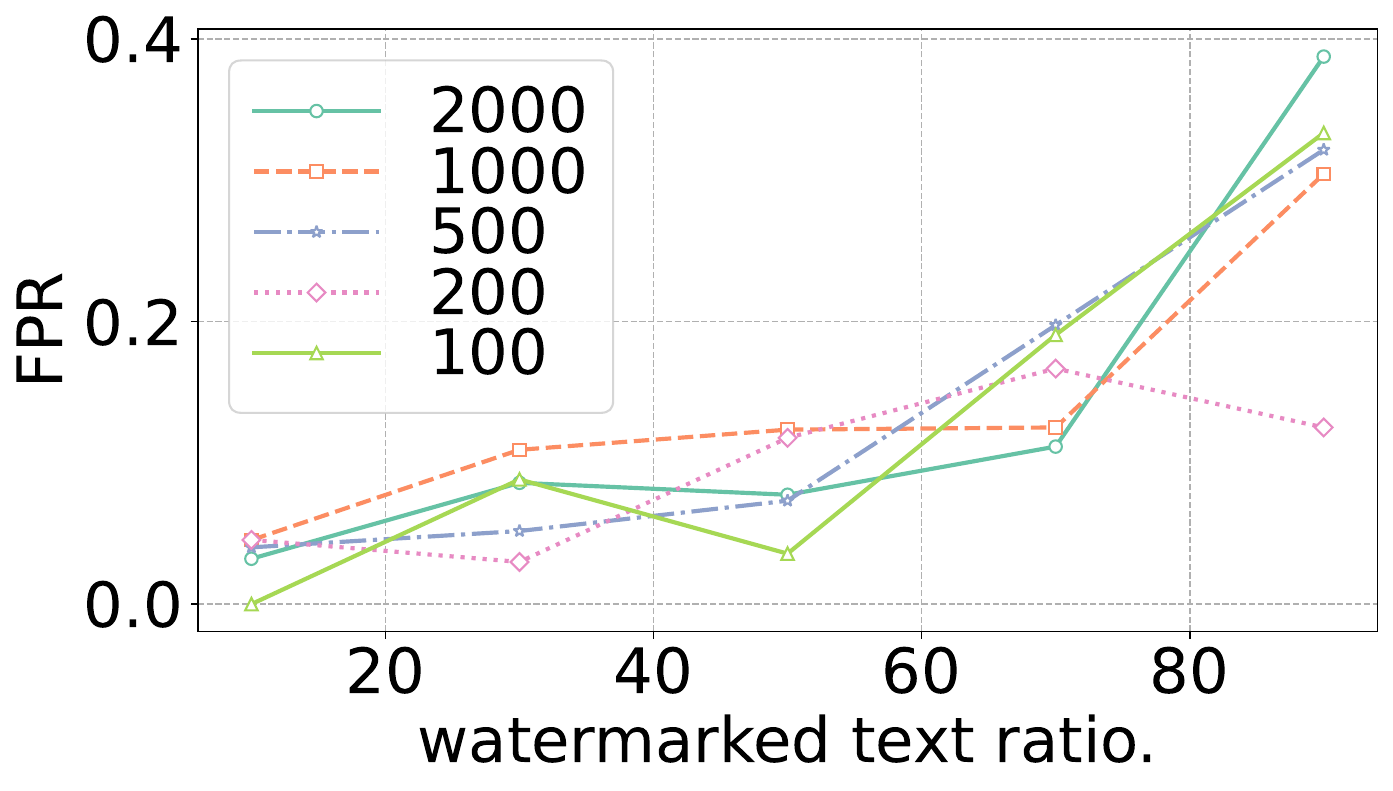}
        \captionof{figure}{Relationship between FPR and sample size. Each curve represents a size.}
        \label{fig:samplecnt}
    \end{minipage}
\end{figure}

\section{Sample Size Experiment}\label{sec:samplesizeexp}
To demonstrate that our experiments used a reasonable sample size and that our results are not sensitive to sample size, we performed an experiment by varying the total sample size within the range of $[100, 200, 500, 1000, 2000]$. As observed in \Cref{fig:samplecnt}, the results did not differ significantly with changes in sample size. This indicates that the sample size we chose is suitable for our current experiment and that our proposed method has good generalizability, not relying heavily on the number of samples. It can be observed that as the ratio of watermarked text increases, the FPR also rises. The underlying reason for this phenomenon is the same as
discussed in \Cref{sec:indicationratio}.

\section{Window Size Parameter $h$}\label{sec:winsizeparam}
We evaluate whether the window size parameter significantly impacts the generation quality, and the results are presented in \Cref{fig:hsim}. It can be observed that as the window size $h$ increases, the text quality scores remain in the range of 69 to 71. In this experiment, no substantial changes in text quality were observed as the window size $h$ varied.

\begin{figure}[H]
    \centering
    \includegraphics[width=0.5\linewidth]{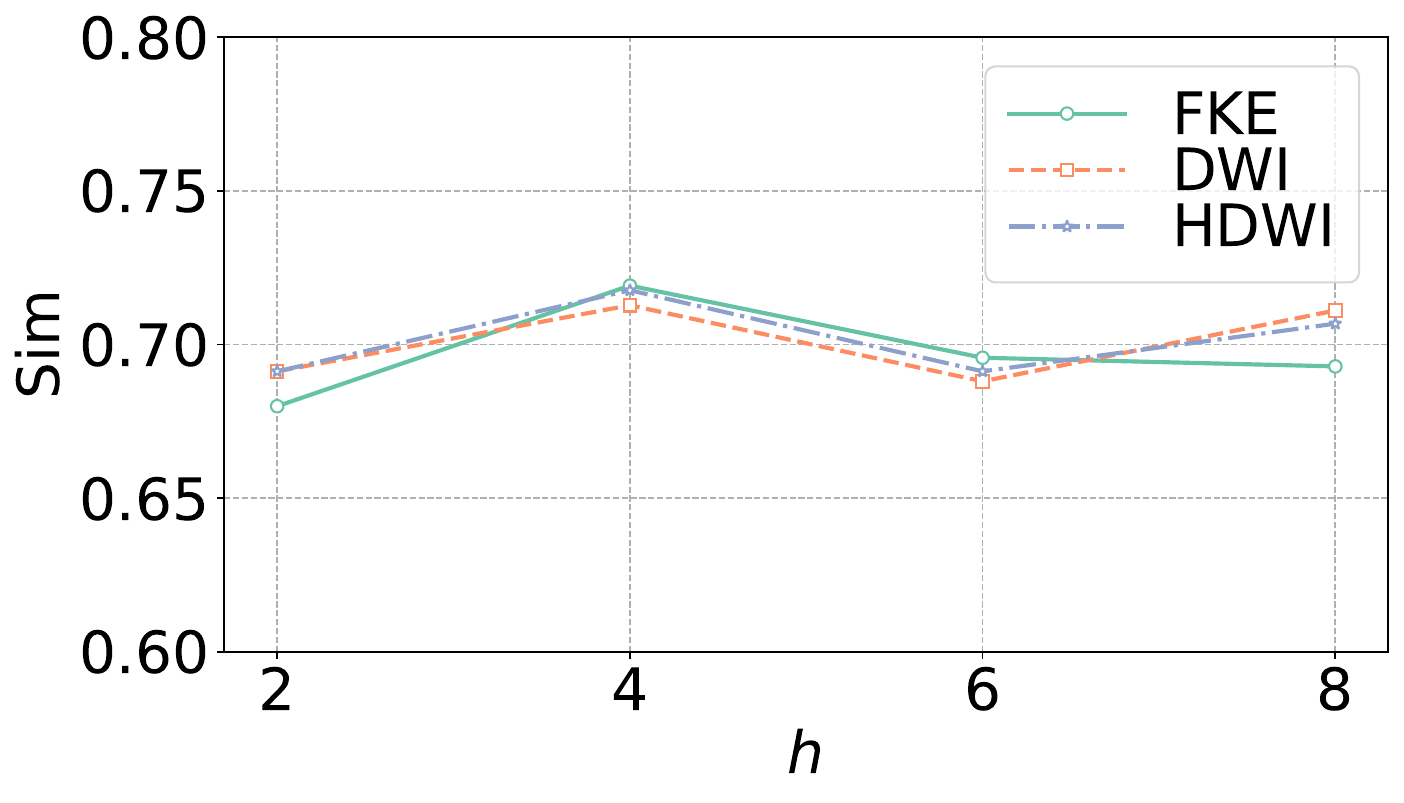}
        \vspace{0em}
        \caption{Text quality with respect to different window sizes $h$.}
        \label{fig:hsim}
\end{figure}

\section{Key Capacity Breakdown Results}\label{sec:keycapacitybreakdown}

We conduct the key capacity breakdown analysis for our method, with the results presented in \Cref{fig:keyexpfig}, providing additional insights into \Cref{sec:keycapacity}. The scores are plotted for various watermarked text ratios, illustrating the performance across diverse scenarios. These plots demonstrate the effectiveness of our proposed method under different watermarked text ratios. For a more detailed discussion, please refer to \Cref{sec:keycapacity}.

\begin{figure}[H]
    \centering
    \begin{minipage}{0.99\linewidth}
    \begin{minipage}{0.49\linewidth}
            \centering
            \includegraphics[width=\linewidth]{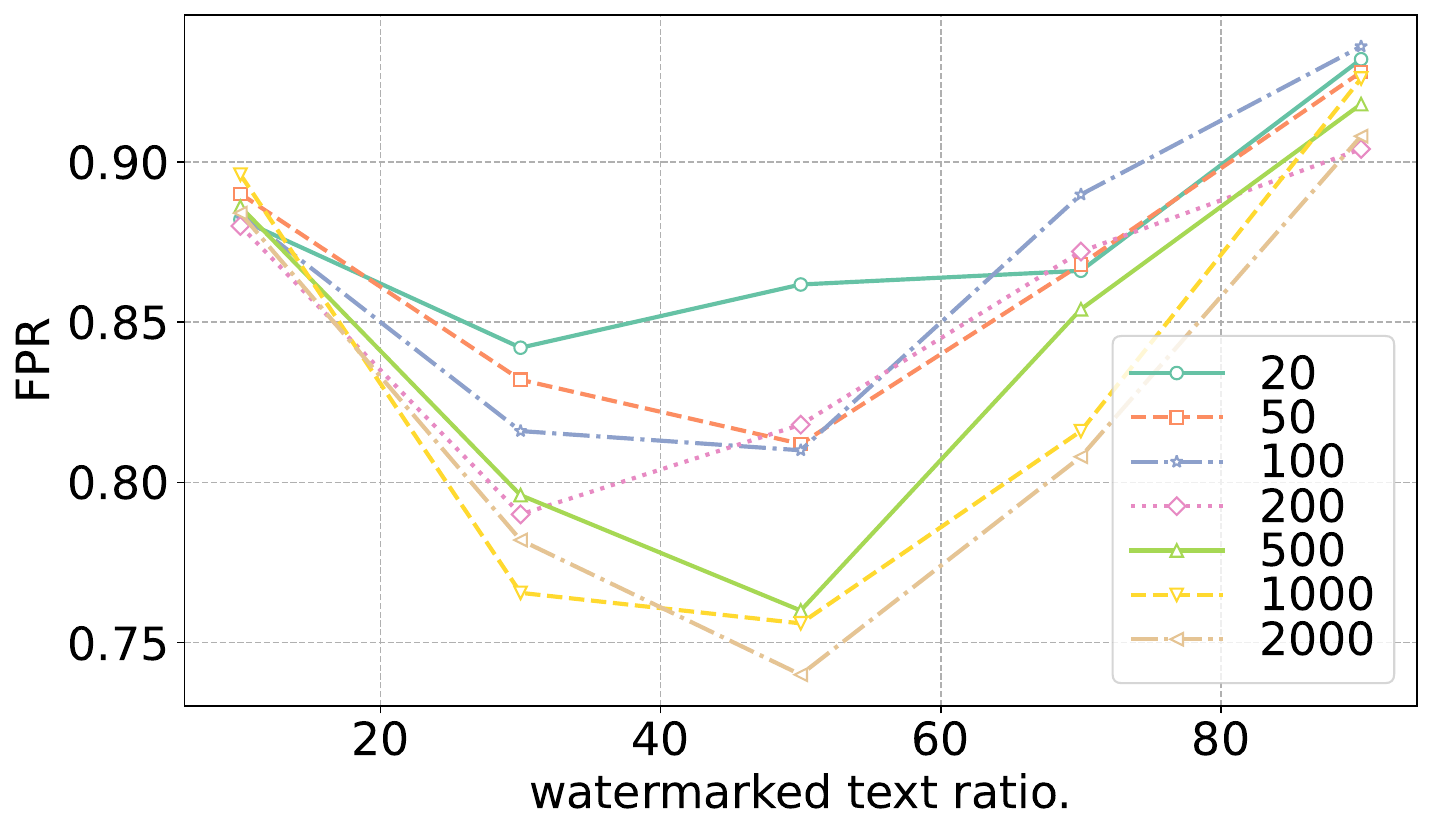}
    \end{minipage}%
    \begin{minipage}{0.49\linewidth}\includegraphics[width=\linewidth]{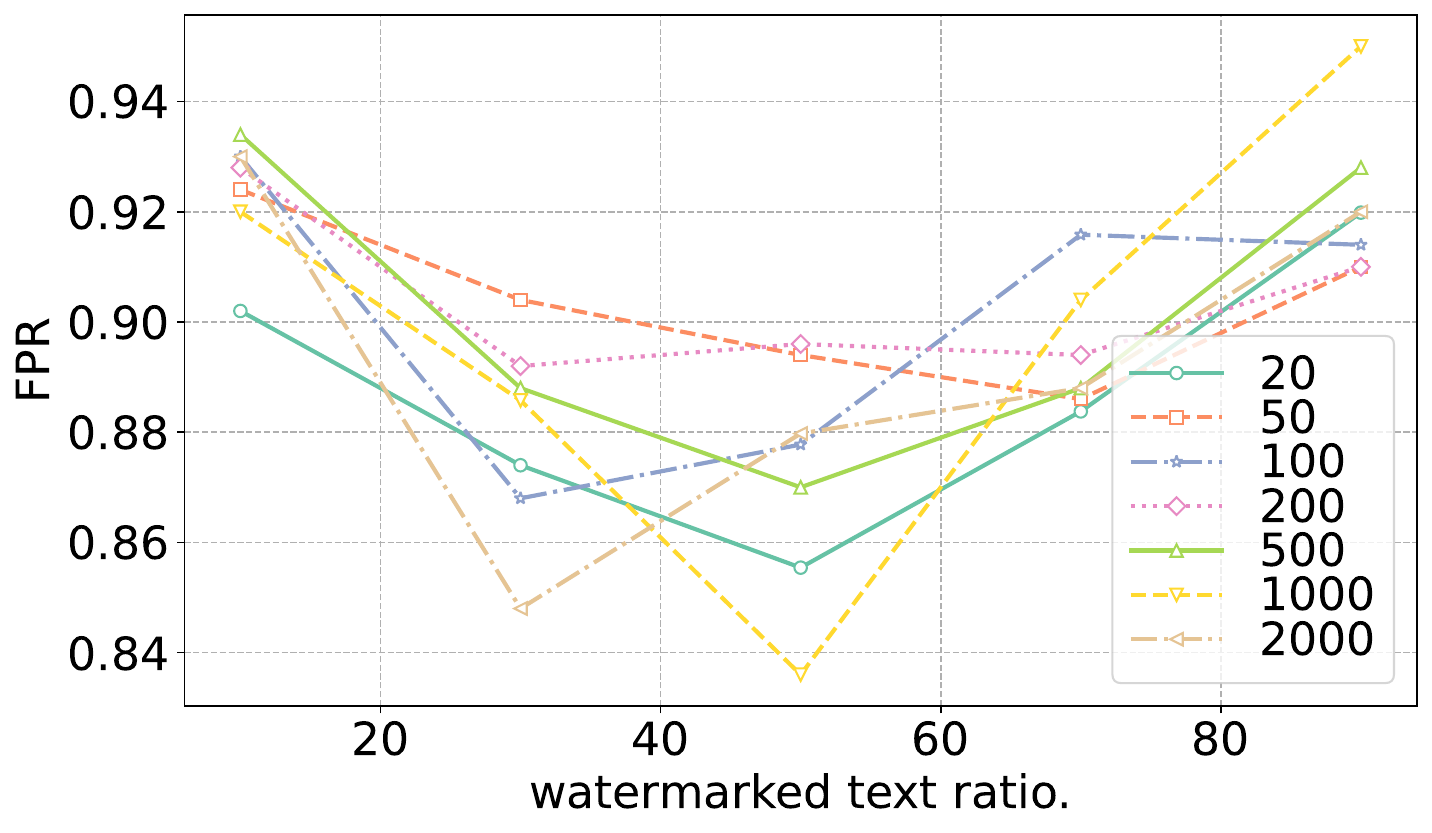}
    \end{minipage}%
        \caption{Key capacity results for FKE (left) and HDW (right). The figures illustrate the relationship between the watermarked text ratio and FPR for varying key capacities ( $K$ ranges from 20 to 2000).}
            \label{fig:keyexpfig}
    \end{minipage}
\end{figure}

\section{Key Capacity Breakdown Results for Dictionary-Based Method}\label{sec:keycapacitybreakdowndictionary}
We further conduct the key capacity breakdown analysis for the dictionary-based method. Similar to \Cref{sec:keycapacitybreakdown}, the results for this method are presented in \Cref{fig:keyexpfigmaryland}, offering additional insights into \Cref{sec:marylandfpr}. The scores are plotted for various watermarked text ratios, illustrating the performance across different scenarios. These findings highlight the behavior and adaptability of the dictionary-based method under varying conditions. For further discussion, refer to \Cref{sec:marylandfpr}.

\begin{figure}[H]
    \centering
    \begin{minipage}{0.99\linewidth}
    \begin{minipage}{0.49\linewidth}
            \centering
            \includegraphics[width=\linewidth]{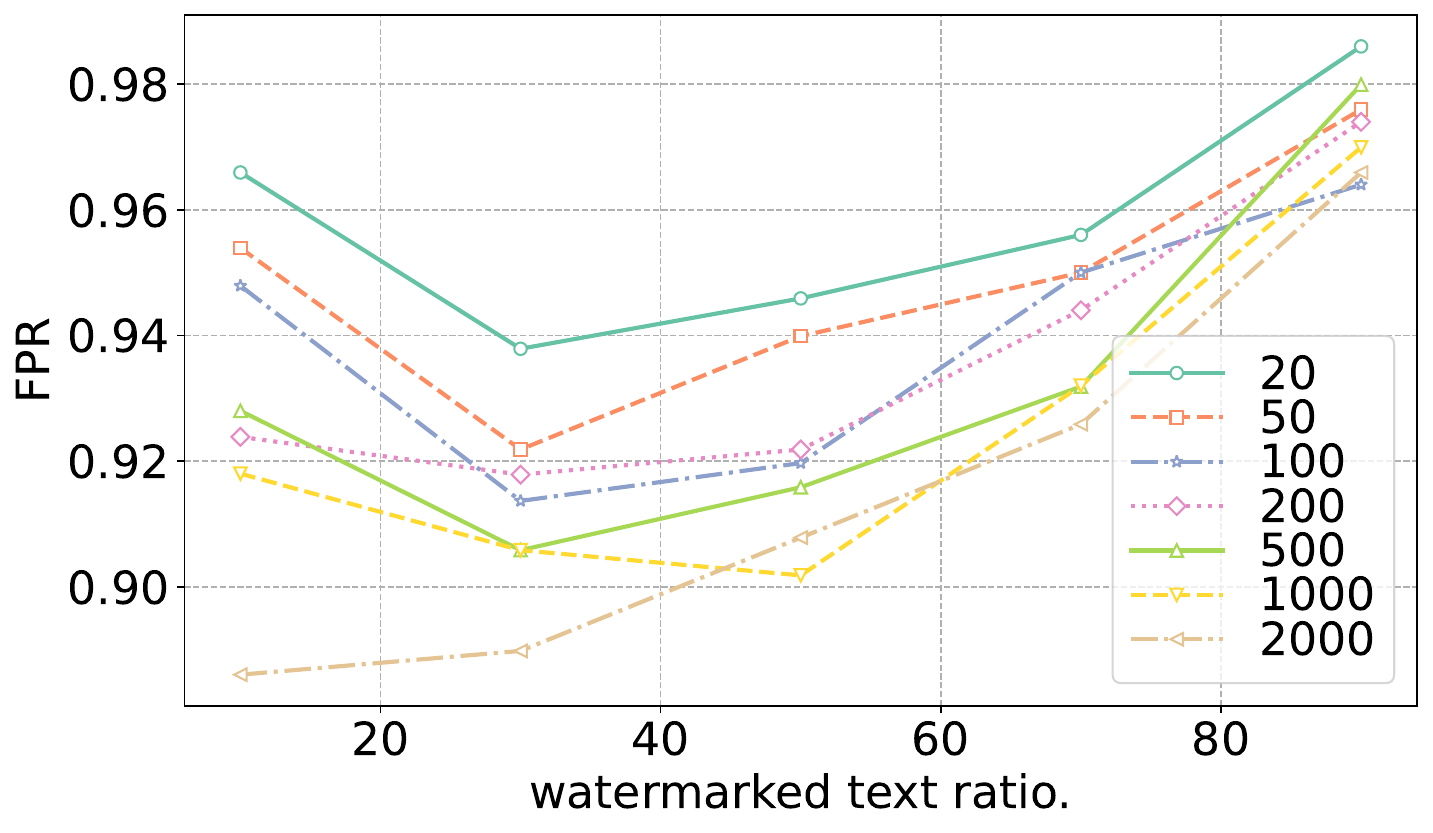}
    \end{minipage}%
    \begin{minipage}{0.49\linewidth}\includegraphics[width=\linewidth]{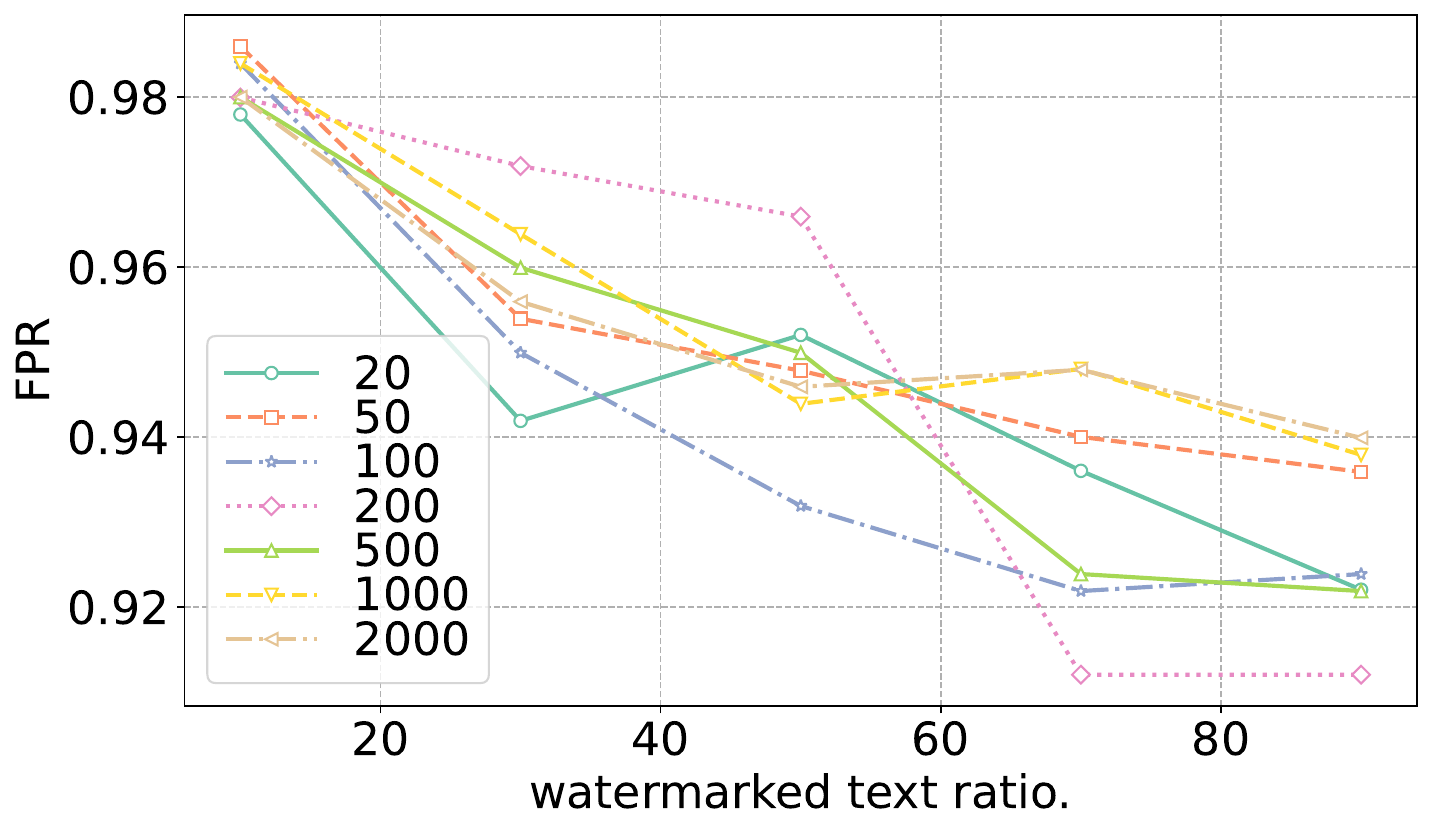}
    \end{minipage}%
        \caption{Dictionary-based model's key capacity results for FKE (left) and HDW (right). All models based on Multi-bit backbone. The figures illustrate the relationship between the watermarked text ratio and FPR for varying key capacities ( $K$ ranges from 20 to 2000).}
            \label{fig:keyexpfigmaryland}
    \end{minipage}
\end{figure}

\section{Key Capacity Experiments for Dictionary-Based Method}\label{sec:keycapacityexpfordicbase}
Similar to the experiments in \Cref{sec:keycapacity}, we conduct key capacity experiments for dictionary-based methods using the Multi-bit approach as our backbone model. The results, presented in \Cref{tab:keyexpmarylandexp}, reveal the following: (1) For traditional dictionary-based methods, the FPR increases as the key capacity grows. This trend aligns with observations from distribution-based methods, further validating the theoretical analysis provided in \Cref{sec:multibit}. (2) Our proposed HDW method effectively addresses this issue, maintaining a consistent FPR scale even as the key capacity increases, thereby demonstrating the robustness and effectiveness of our approach.

\begin{table}[H]
    \centering
    \begin{minipage}{0.4\textwidth}
        \centering
        \begin{tabular}{@{~}l@{~}@{~}l@{~}@{~}l@{~}@{~}l@{~}}
        \toprule
         & Accu-I$\uparrow$ & Accu-O$\uparrow$ & FPR$\downarrow$ \\
        \midrule
        20 & 0.958 & 0.937 & 0.0647 \\
        50 & 0.948 & 0.916 & 0.096 \\
        100 & 0.939 & 0.907 & 0.121 \\
        200 & 0.936 & 0.898 & 0.0928 \\
        500 & 0.932 & 0.889 & 0.13 \\
        1000 & 0.926 & 0.884 & 0.147 \\
        2000 & 0.915 & 0.874 & 0.141 \\
        \bottomrule
        \end{tabular}
    \end{minipage}%
    \quad\quad\quad
    \begin{minipage}{0.29\textwidth}
        \centering
        \begin{tabular}{@{~}l@{~}@{~}l@{~}@{~}l@{~}@{~}l@{~}}
        \toprule
         & Accu-I$\uparrow$ & Accu-O$\uparrow$ & FPR$\downarrow$ \\
        \midrule
        20 & 0.946 & 0.923 & 0.0426 \\
        50 & 0.953 & 0.925 & 0.0277 \\
        100 & 0.942 & 0.905 & 0.0293 \\
        200 & 0.948 & 0.918 & 0.0189 \\
        500 & 0.947 & 0.909 & 0.036 \\
        1000 & 0.956 & 0.897 & 0.0391 \\
        2000 & 0.954 & 0.903 & 0.0367 \\
        \bottomrule
        \end{tabular}
    \end{minipage}
    \caption{Key capacity results for FKE (left) and HDW (right) methods at varying key capacities (20 to 2000). All models use Multi-bit method as backbone.}
    \label{tab:keyexpmarylandexp}
\end{table}

\end{document}